%% file: main.tex
\documentclass[hidelinks,onefignum,onetabnum]{siamart251216}

\input{ex_shared}
\usepackage{braket}
\usepackage[numbers,sort&compress]{natbib}
\ifpdf
\hypersetup{
  pdftitle={State preparation},
  pdfauthor={allstars}
}
\fi

\begin{document}

\maketitle

% REQUIRED
\begin{abstract}
Quantum state preparation is a fundamental primitive in quantum algorithms for encoding classical data into quantum amplitudes. We compare the cost of preparing general $n$-qubit states with real amplitudes using two common paradigms: rotation-based methods, based on controlled rotations, and sampling-based methods, based on a structured representation of the target state. Although these approaches are often theoretically compared using CNOT count and $T$-count, their relative performance in total gate count remains less well understood practically. We compare representative rotation-based and sampling-based methods using $T$-count and total gate count, and analyze how compilation overhead affects their relative performance. We also develop a software package for compiling state preparation circuits, designed as a practical subroutine for more general quantum computations. Numerical experiments on resource states and quantum states related to quantum chemistry, condensed matter physics, and simulation via Magnus expansion over a range of target accuracies $\epsilon$ support the analysis. Our results show that sampling-based methods achieve asymptotically lower $T$-count and retain an overall advantage after accounting for total gate count and compilation overhead.
\end{abstract}

% REQUIRED

% REQUIRED
% \SZ{Somewhere we have to justify that we only consider real coefficients in this work.} \DL{We add a short discussion in sec 2.}
% \SZ{Todo friday: 1.run time analysis 2. for magnus, fix precision, run k vs count for sampling}
% \DL{In general, we may need to interpret the numerical experiments a bit. I added a few comments.}\SZ{Working on it.}
%\DL{change D to L}

\section{Introduction}
Quantum state preparation is a fundamental subroutine in quantum computing, whose goal is to encode classical data into a quantum state for subsequent algorithmic tasks. In quantum simulation, quantum state preparation is often used to construct prepare or amplitude oracles within the block-encoding framework~\cite{Low2017Optimal,Gilyen2019QSVT,liu2025efficient,camps2024explicit,liu2025block}, serving as essential primitives for the overall simulation algorithm. In addition, state preparation is a necessary step in quantum algorithms for solving linear systems~\cite{harrow2009quantum,ambainis2012variable,childs2017quantum,clader2013preconditioned,wossnig2018quantum,costa2022optimal,low2026quantum} and partial differential equations~\cite{montanaro2016quantum,costa2019quantum,childs2020quantum,childs2021high,jin2022quantum,liu2023efficient}. In both cases, one typically must first prepare a quantum state that encodes the right-hand side 
of the equations in its amplitudes. More generally, any quantum computation that requires loading classical information directly into amplitudes depends on a quantum state preparation subroutine.

%In addition, the state preparation is also a necessary step for solving linear system $Ax=b$ and partial differential equations $\mathcal{L}u=v$ on quantum computer. For these problems, it is always necessary to first prepare a quantum state corresponding to the right-hand-side for the equations. We note that any quantum computation tasks requiring direct loading classical information encoded in amplitude requires the subroutine for state preparation.

Quantum state preparation algorithms~\cite{ross2016optimal,gosset2026quantum,
zhang2022quantum,yuan2023optimal,gui2024spacetime,
Low2024tradingtgatesdirty,wang2024quantum} use sequences of single-qubit and
two-qubit gates to load classical data into a quantum state. To compare the
resource costs of different state-preparation algorithms, we express these
circuits in a common elementary gate set. It is well known that the gate set
$\{\mathrm{CNOT}, H, T\}$ is approximately universal for quantum
computation~\cite{nielsen2010quantum}, meaning that the gates appearing in
state-preparation circuits can be approximated to arbitrary precision using
gates from this set. In fault-tolerant quantum
computing~\cite{gottesman2024surviving}, however, these elementary gates do not
have comparable implementation costs. Clifford gates, generated for example by
$\mathrm{CNOT}$, $H$, and $S=T^2$, are usually relatively inexpensive, whereas
the non-Clifford $T$ gate often dominates the resource
overhead~\cite{BabbushGidneyEtAl18}. Fault-tolerant implementations of a $T$ gate commonly rely on an algorithm called magic-state distillation~\cite{bravyi2005universal,bravyi2012magic,jones2013multilevel,haah2017magic,hastings2018distillation,campbell2017roads,litinski2019magic}, which can require an overhead scaling as $\log^{\gamma}(1/\epsilon)$ to achieve a target error rate $\epsilon$ for some positive $\gamma$. As a result, realizing high-fidelity $T$ gates can be substantially more costly than implementing Clifford gates. For this reason, the performance of many fault-tolerant quantum algorithms is frequently assessed by their $T$-count (and sometimes $T$-depth): fewer $T$ gates generally implies lower overall resource requirements.

%via amplitude encoding

%State preparation algorithms use a sequence of single qubit gate and two qubit gate to load the classical information with amplitude format. It is well-known that $\text{CNOT}, \text{H}, \text{S}$ gates and $T$ gate constitute a universal gate set which is enough for general quantum computing. However, the fault tolerant implementation of these gates does not have comparable cost. For a long time, the fault tolerant implementation of $T$ gate, a specific type of non-Clifford gate, has a implementation overhead $\log^{\gamma}(1/\epsilon)$ based on a method called magic state distillation. A high performance fault tolerant implementation of $T$ gate may require a larger number of gate overhead compared to fault-tolerant implementation of Clifford gate. For a long time, the performance of a quantum algorithm is measured with the number of $T$ gate. Less $T$ gate is used, the better is the performance of quantum algorithm.

Recent advances in magic-state cultivation~\cite{gidney2024magic,vaknin2025efficient,vaknin2026high,chen2026efficient} provide a route to fault-tolerantly implementing the $T$ gate with reduced resource requirements. In particular, for target error rates around $\epsilon \approx 10^{-7}$, resource estimates suggest that high-fidelity $T$ gates may be realized at a cost that is closer to that of Clifford gates, although $T$ gates are still expected to remain more expensive in practice for higher precision. At the same time, a series of works has proposed magic-state distillation protocols with constant overhead~\cite{wills2025constant,nguyen2025good,cervia2025magic}, as well as approaches to fault-tolerant quantum computation with constant overhead~\cite{gottesman2013fault,xu2024constant,Pablo2025ConstantOverheadFB}.

% Recent work on magic state cultivation gives one way to fault tolerantly implement $T$ gate with reduced resources. Indeed, for $\epsilon \approx 10^{-7}$, it is expected that the overhead for high performance $T$ gate can be implemented with similar cost compared to Clifford gate although high performance $T$ gate are still significantly more expensive than Clifford gate. On the same time, a sequence of work has been proposing magic state distillation with constant overhead~\cite{wills2025constant,nguyen2025good,cervia2025magic} or fault tolerant quantum computing with constant overhead~\cite{gottesman2013fault}.

These developments motivate a closer examination of algorithm selection for state preparation. In both the early fault-tolerant regime, where non-Clifford gates incur costs similar to those of Clifford gates, and the constant-overhead fault-tolerant regime, a low $T$-count will no longer serve as the primary criterion for evaluating the performance of a quantum algorithm. Instead, the total gate count, which acts as a proxy for the overall space-time volume, provides a more accurate indicator of algorithmic efficiency.

State preparation methods for a general quantum state fall into two main classes. The first class constructs an amplitude-encoded target state using a sequence of (possibly controlled) rotation gates, with rotation angles determined by the input data. The second class first loads or computes a binary (computational-basis) representation of the data and then applies a sampling-based procedure to convert this representation into the desired amplitude-encoded quantum state. Although sampling-based methods have been theoretically shown to achieve lower $T$-gate complexity than rotation-based methods~\cite{BabbushGidneyEtAl18,Low2024tradingtgatesdirty,ross2016optimal}, and some comparisons exist for the NISQ regime where CNOT gate counts are paramount~\cite{AlonsoLinaje2025QuantumCF}, a thorough comparison of these two approaches in the early and constant-overhead fault-tolerant regimes remains lacking.

%Although it has been shown that the T gate complexity for sampling based method is smaller the the rotation based method~\cite{BabbushGidneyEtAl18,Low2024tradingtgatesdirty,ross2016optimal} and some comparison has been done for NISQ regime where the number of CNOT gate matters most~\cite{AlonsoLinaje2025QuantumCF}, the comparison between the two type method was not thoroughly done for the early fault tolerant regime and constant overhead fault tolerant regime.  

Our contributions are threefold. First, we develop a software package for constructing, compiling, and comparing the costs of representative dense and sparse rotation-based circuits against sampling-based state preparation, \QROM and \SelectSwap. Second, we compare these approaches using both compiled $T$-count and compiled total gate count, explicitly accounting for the single-qubit synthesis overhead incurred by rotation-based methods. Third, we benchmark the resulting circuits on structured resource states and application-motivated examples from quantum chemistry, condensed-matter physics, and Magnus expansion simulation.

 The remainder of the paper is organized as follows. In \Cref{sec:alg}, we review both state preparation methods along with their associated compilation techniques. Next, we provide a theoretical comparison of the two approaches in \Cref{sec:theory}, followed by numerical experiments evaluating their performance in \Cref{sec:numerical}. Finally, \Cref{sec:conclusions} provides further discussion and our concluding remarks.
%In the paper, we compare the rotation based and sampling based methods for state preparation of a general quantum state with real amplitude. We review the two types of state preparation methods and compilation technique in \Cref{sec:alg}. We theoretically compare the sampling based method with rotation based method in \Cref{sec:theory}. Numerical experiments for comparing the two type of methods are shown in \cref{sec:resource}, More discussion and conclusion are incuded in \cref{sec:conclusions}.

%We numerically study the performance of several state-preparation algorithms. We also seek to develop practical, compilation-independent criteria that enable informed method selection prior to running the compiler.

%\textbf{Nothing in this paper is about improving asymptotic scaling. This paper is from constant and for constant.}
%\section{Preliminary}

% \section{Main results}
% \label{sec:main}

\section{State preparation algorithms}
\label{sec:alg}
We write the target state as
\begin{equation}
    \ket{\psi} = \sum_{j=0}^{L-1} \alpha_j \ket{j},
    \qquad \alpha_j \in \mathbb{R},
    \qquad \sum_{j=0}^{L-1} \alpha_j^2 = 1,
    \qquad L = 2^n,
\end{equation}
where $n$ is the number of qubits.
% For a comprehensive introduction to the bra-ket notation used throughout this paper, we refer the reader to~\cite{nielsen2010quantum}. Throughout this paper, and throughout the codebase used for the experiments, we restrict attention to \emph{real}-amplitude states. Since real-amplitude states arise more commonly in quantum computing, we focus on this setting. For states with complex-amplitude we can decompose it to a real-amplitude state preparation problem plus a phase oracle synthesis problem, therefore. The corresponding complexity separation for complex-amplitude states is expected to be similar.
For a comprehensive introduction to the bra-ket notation used in this work, we refer the reader to~\cite{nielsen2010quantum}. Both in this paper and the accompanying codebase, we restrict our attention to \emph{real}-amplitude states, which are ubiquitous in quantum computing applications. Because preparing a complex-amplitude state can be decomposed into a real-amplitude state preparation step followed by phase oracle synthesis, we expect the complexity separation between methods to remain similar for complex states.
It is therefore convenient to introduce the probability distribution
\begin{equation}
    p_j = \alpha_j^2,
    \qquad j \in [L] := \{0,1,\ldots,L-1\},
\end{equation}
with sparsity $m = |\{j : \alpha_j \neq 0\}|$.
This notation makes it clear that the two families considered in this work solve slightly different tasks.
Rotation-based methods aim to prepare the full state $\ket{\psi}$, whereas the sampling-based methods we use, based on alias sampling, are designed to reproduce the target distribution $p$ on an address register.
In particular, the sampling-based family depends only on the distribution $p$ and therefore does not preserve amplitude signs.
We account for this distinction explicitly in the fidelity metrics introduced in~\cref{par:output-metrics}.
% In this section we consider the generic task of preparing
% \begin{equation}
%     \ket{\psi_x} = \sum_{j=0}^{D-1} x_j \ket{j}, \qquad x\in \mathbb{C}^D,\quad \|x\|_2=1,\quad D=2^n.
% \end{equation}
% From the perspective of logical resource estimation, the most important design choice is how the classical description of $x$ is converted into amplitudes. 
% %
% The two broad families considered in this paper are lookup table-based methods and rotation-based methods.
% %
% Although both families have worst-case complexity that scales linearly in $D$ up to polylogarithmic factors, their compiled costs can differ substantially.
% %
% Lookup table-based methods replace a large fraction of the arbitrary-angle synthesis by reversible data-access primitives, while rotation-based methods keep the circuit structure simple but expose many independently synthesized single-qubit rotations.
%

\subsection{Sampling-based state preparation}
\label{sec:alias}
The sampling-based state preparation method originates from techniques used for the classical sampling problem. In particular, alias sampling has been proposed to generate a random variable following a given probability distribution~\cite{vose1991linear}. At a high level, the alias sampling method rewrites the target probability distribution $p$ using a uniform draw over bins together with a two-branch correction step.
A classical preprocessing stage constructs, for the $j$-th bin, a threshold $\tau_{j} \in [0,1]$ and an alias label $\operatorname{alias}_j \in [L]$ such that if $J$ is sampled uniformly from $[L]$ and $u$ is sampled uniformly from $[0,1)$, then the random variable
\begin{equation}
\label{eq:alias-classical}    J_{\mathrm{out}} =
    \begin{cases}
        J, & u < \tau_J, \\
        \operatorname{alias}_{J}, & u \geq \tau_J,
    \end{cases}
\end{equation}
has distribution $p$.

In a similar spirit, this classical sampling method can be adapted to quantum computing, and in particular to the quantum state preparation problem. Such a generalization is unsurprising, since every quantum state admits a natural interpretation as a probability distribution. Before discussing how the alias sampling method prepares a quantum state, we first review the two reversible primitives it relies on: a quantum read-only memory for loading the precomputed table, and a quantum comparator for performing the threshold check coherently.

%The coherent quantum task is then to prepare a superposition over $J$, load the corresponding classical data ($ \tau_{J}$ and $ \text{alias}_{J}$), perform the comparison coherently, and route amplitude to the correct output branch.

% The central reversible primitive is a quantum read-only memory (\QROM)~\cite{BabbushGidneyEtAl18}, which we write as
% \begin{equation}
%     \QROMOp_{\mathcal{D}} : \ket{j}\ket{z} \mapsto \ket{j}\ket{z \oplus \mathcal{D}_j},
% \end{equation}
% where $\mathcal{D}_j$ denote the $j$-th data in the form of a binary string.
% In the fault-tolerant setting, \QROM replaces a large number of unrelated arbitrary-angle syntheses by structured reversible data access, and this is precisely why it is attractive for state-preparation subroutines.
% %
% The related \SelectSwap~\cite{Low2024tradingtgatesdirty} construction introduces a tunable space--time tradeoff, parameterized by a block size $\lambda$, by combining batched data loads with a controlled swap network.
% This type of arithmetic-free state preparation is also closely connected to the black-box sampling paradigm of Ref.~\cite{SandersLowEtAl19}. 
% %
% We refer readers to Ref.~\cite{ZhuSundaramLow25} for a comprehensive review of quantum lookup tables.

The central reversible primitive in our construction is the quantum read-only memory (\QROM)
\cite{BabbushGidneyEtAl18}, written as
\begin{equation}
    \QROMOp_{\mathcal{D}}:\ket{j}\ket{z}\longmapsto \ket{j}\ket{z\oplus \mathcal{D}_j},
\end{equation}
where $\mathcal{D}_j$ denotes the $j$-th data item encoded as a binary string. 
In the fault-tolerant regime, \QROM replaces many instances of unrelated, high-cost arbitrary-angle synthesis by a single, structured reversible data-access primitive. This structured access is particularly advantageous for state-preparation subroutines, where one repeatedly loads tabulated quantities into a coherent superposition.

A closely related approach is the \SelectSwap construction \cite{Low2024tradingtgatesdirty}, which introduces an adjustable space--time tradeoff by batching data loads and routing them through a controlled swap network. The tradeoff is parameterized by a block size $\lambda$, allowing one to interpolate between more qubit-intensive and more time-intensive implementations. 
This arithmetic-free state preparation perspective is also closely connected to the black-box sampling paradigm of Ref.~\cite{SandersLowEtAl19}. 
For a detailed discussion of quantum lookup-table techniques, we refer the reader to Ref.~\cite{ZhuSundaramLow25}.
%\DL{start from here}

A quantum comparator is a subroutine that determines the ordering of two integers in binary form. It is defined by its action on computational basis states:
$$\mathsf{COMP}\ket{x}\ket{y}\ket{0}=\begin{cases}
\ket{x}\ket{y}\ket{1}, & \text{if } y \geq x,\\[4pt]
\ket{x}\ket{y}\ket{0}, & \text{otherwise},
\end{cases}$$
writing the predicate $(y \geq x)$ into the ancilla. This converts a classical comparison into a coherent quantum operation that can be embedded in larger algorithms. Efficient implementations scan the binary representations from the most significant bit downward, proceeding to lower bits only when the higher bits do not yet determine the ordering---mirroring classical comparison logic, but implemented reversibly.

We now review the coherent implementation of alias sampling proposed in~\cite{BabbushGidneyEtAl18}. The classical uniform sampling step is replaced by the preparation of the uniform superposition
$\frac{1}{\sqrt{L}}\sum_{j=0}^{L-1}\ket{j}$.
Using a data-lookup table, we load the pair of integers $\text{alias}_j \in [L]$ and $\text{keep}_j \in [2^b]$ (in binary) associated with each index $j$, obtaining
\begin{equation}
    \frac{1}{\sqrt{L}}\sum_{j=0}^{L-1}\ket{j}\ket{\text{alias}_j}\ket{\text{keep}_j}.
\end{equation}
Applying $b$ Hadamard gates to a fresh register produces
\begin{equation}
\label{eq:hadamard-register}
    \frac{1}{\sqrt{2^b L}}\sum_{j=0}^{L-1}\ket{j}\ket{\text{alias}_j}\ket{\text{keep}_j} \sum_{\sigma=0}^{2^b-1}\ket{\sigma}.
\end{equation}
A quantum comparator and $b$ controlled swap gates are then used to swap $\ket{\text{alias}_j}$ and $\ket{j}$ whenever $\sigma \geq \text{keep}_j$, resulting in the state
\begin{equation}
    \sum_{j=0}^{L-1} \tilde{\alpha}_{j} \ket{j} \ket{\text{temp}_j}, \qquad \tilde{\alpha}_j^2=\frac{\text{keep}_j+\sum_{k:\  \text{alias}_k=j}(2^b-\text{keep}_k)}{2^{b}L},
\end{equation}
where $\ket{\mathrm{temp}_j}$ denotes a normalized state collecting the remaining ancilla registers. The values $\text{alias}_j$ and $\text{keep}_j$ are precomputed classically---an efficient procedure---so that $\tilde{\alpha}_j$ closely approximates the target amplitude $\alpha_j$. A rigorous error analysis for a representative state preparation problem is given in Theorem~\ref{thm:single-qubit-lookup}. In particular, we note that the garbage state $\ket{\text{temp}_j}$ does not influence the construction of the input model for quantum simulation.

\subsection{Rotation-based state preparation}
Rotation-based methods act directly on amplitudes.
For real states, they can be viewed as recursively applying $R_y$ rotations conditioned on previously prepared qubits, in the spirit of uniformly controlled rotation decompositions and related state-preparation constructions from the literature~\cite{mottonen2005transformation,plesch2011state}.
The two variants considered in this work differ in which measure of problem size they reduce at each iteration.

\paragraph{Dense} %\HW{Just a side note, these names, qubit reduction and cardinality reduction, are from my DATE'24 paper, not sure if they are common knowledge...}
The dense strategy reduces the number of active qubits~\cite{mozafari2022efficient}.
At each step one chooses a pivot qubit, groups amplitudes according to the values of the remaining qubits, and applies a uniformly controlled $R_y$ on the pivot so that the two branches associated with that qubit are merged into a single effective branch.
After this merge, the problem depends on one fewer active qubit.
This viewpoint is natural when weight is spread broadly across the computational basis.

\paragraph{Sparse}
The sparse strategy~\cite{gleinig2021efficient} reduces the number of occupied basis states.
Instead of eliminating one qubit at a time, it repeatedly chooses two nonzero amplitudes, aligns them so that they differ on only one qubit, and then applies a controlled $R_y$ to merge them.
Each iteration reduces the support size from $m$ to $m-1$.
This is advantageous when $m \ll L$, because the work scales with the occupied support rather than with the full Hilbert-space dimension~\cite{LiLuo2025Sparse}.

% \subsection{Rotation-based method}
% Rotation-based methods construct the amplitudes of $\ket{\psi_x}$ directly through a series of
% single-qubit rotations controlled by previously prepared qubits. 
% %
% For generic complex states, the
% amplitude magnitudes and phases are usually handled separately: one first prepares the correct magnitude profile using controlled $R_y$ rotations and then applies phases using controlled $R_z$ rotations or an equivalent diagonal ~\cite{mottonen2005transformation,plesch2011state}.
% %
% The appeal of this family is conceptual simplicity: after a classical preprocessing step that computes rotation angles, the quantum circuit follows a fixed recursive pattern.
% %
% The two variants considered here differ in what they simplify at each step. The dense method reduces the number of active qubits, while the sparse method reduces the number of nonzero amplitudes.

% \paragraph{Dense}
% The dense method is a qubit-reduction strategy. One chooses a pivot qubit and groups the amplitudes according to the values of the remaining qubits. A single $R_y$ rotation, or more generally a uniformly controlled $R_y$ rotation, is then used so that the amplitudes on the $\ket{0}$ and $\ket{1}$ branches of the pivot combine into a single effective branch. After this step, the state depends on one fewer active qubit, and the same procedure is repeated. This is the natural approach when amplitude is spread broadly across the computational basis, since it follows the recursive structure of the full state vector.

\subsection{Implementation details}
This subsection records the conventions used in the implementation and should be read as the implementation-level specification for the numerical experiments.

% \paragraph{State representation}
% The target input is stored as a sparse dictionary $j \mapsto \alpha_j$ of real amplitudes.
% The rotation-based routines manipulate this sparse representation directly, while the sampling-based routine first converts it to the probability vector $p_j = \alpha_j^2$ and then builds an alias table from $p$.\DL{what do we want to say here?}

\paragraph{Dense routine}
The dense routine (\cref{alg:dense}) implements a qubit-reduction loop.
The pivot qubit is not arbitrary: it is chosen to minimize the imbalance between the number of occupied basis states on its $0$ and $1$ branches.
For the chosen pivot $q$, the code collects the corresponding branch amplitudes $(a_0,a_1)$, forms an angle table
\begin{equation}
    \theta = 2\operatorname{atan2}(a_1,a_0),
\end{equation}
prunes any controls on which the angle table is constant, and emits the resulting uniformly controlled $R_y$ through a recursive demultiplexing circuit.
If the qubit-reduction loop stalls before reaching a single basis state, the implementation falls back to the sparse routine on the residual state.
In the no-fallback case, the emitted dense circuit consists of $R_y$, CNOT, and Pauli-$X$ gates.
\begin{algorithm}[htbp]
\caption{Dense routine}
\label{alg:dense}
\footnotesize
\begin{algorithmic}
\STATE \textbf{Input/Output:} $\alpha=(\alpha_j)_{j=0}^{L-1}\in\mathbb{R}^L$, $L=2^n$; return $U_{\mathrm{dense}}$ with $U_{\mathrm{dense}}\ket{0^n}=\sum_{j=0}^{L-1}\alpha_j\ket{j}$.
\STATE Initialize $s \gets \{j \mapsto \alpha_j : \alpha_j \neq 0\}$ and an empty gate list $\mathcal{G}$.
\WHILE{the support-qubit set of $s$ has size greater than $1$}
\STATE Choose a pivot qubit $q$ that minimizes the imbalance between the occupied $q=0$ and $q=1$ branches, and let $C$ be the remaining support qubits.
\FOR{each control pattern $y \in \{0,1\}^{|C|}$}
\STATE Compute the branch amplitudes $a_0(y)$ and $a_1(y)$ and set $\theta_y \gets 2\operatorname{atan2}(a_1(y),a_0(y))$.
\ENDFOR
\STATE Remove any control qubit on which the angle table $\{\theta_y\}_y$ is constant, append the resulting $R_y$ or uniformly controlled $R_y$ gate on qubit $q$ to $\mathcal{G}$, and replace each branch pair by the merged amplitude $\sqrt{a_0(y)^2+a_1(y)^2}$ to obtain the reduced state $s$.
\ENDWHILE
\STATE Prepare the residual state directly if $s$ is a singleton and otherwise by the sparse routine.
\STATE Apply the recorded gates in $\mathcal{G}$ in reverse order and return the resulting circuit.
\end{algorithmic}
\end{algorithm}

\paragraph{Sparse routine}
The sparse routine (\cref{alg:sparse}) implements a cardinality-reduction loop.
It greedily isolates a pair of occupied basis states by fixing informative bit values, aligns that pair by a sequence of CNOTs so that the two states differ on only one qubit, and then applies a multi-controlled $R_y$ on the differing qubit to merge the pair.
Negative controls are implemented by surrounding the relevant control qubits with Pauli-$X$ gates.
At the logical level, this routine may therefore contain $R_y$, $\mathrm{CR}_y$, and higher-controlled $R_y$ gates.

\begin{algorithm}[htbp]
\caption{Sparse routine}
\label{alg:sparse}
\footnotesize
\begin{algorithmic}
\STATE \textbf{Input/Output:} $\alpha=(\alpha_j)_{j=0}^{L-1}\in\mathbb{R}^L$, $L=2^n$; return $U_{\mathrm{sparse}}$ with $U_{\mathrm{sparse}}\ket{0^n}=\sum_{j=0}^{L-1}\alpha_j\ket{j}$.
\STATE Initialize $s \gets \{j \mapsto \alpha_j : \alpha_j \neq 0\}$ and an empty gate list $\mathcal{G}$.
\WHILE{$|s|>1$}
\STATE Greedily fix informative bit values until one occupied basis index $i_0$ remains, remove the last fixed bit, and choose a compatible second index $i_1$.
\STATE Let $q$ be the differing bit, let $\Gamma$ be the remaining fixed-bit conditions, use CNOTs controlled by $q$ to align the pair, and append the resulting multi-controlled $R_y$ to $\mathcal{G}$.
\STATE Replace the aligned pair by a single occupied basis state to update $s$.
\ENDWHILE
\STATE Prepare the remaining basis state from $\ket{0^n}$ by Pauli-$X$ gates.
\STATE Apply the recorded gates in $\mathcal{G}$ in reverse order and return the resulting circuit.
\end{algorithmic}
\end{algorithm}

\paragraph{Sampling-based state preparation}
%\DL{The notation is not consistent here. For alias sampling, the notation used is $u$ and $\tau_J$.} 

%\DL{1. how is control multi X being implemented right now 2. how is multi-control not implemented right now. 3. The difference between the QROM and SELECT-SWAP, point out the high level difference should be good like the balance between qubits usage and T gate count.
%}

The sampling-based state preparation routine (\cref{sec:alias}) first constructs a $b$-bit alias table from $p$.
For each bin $j$, the code stores an integer $\text{keep}_j \in \{0,\ldots,2^b-1\}$, obtained as $\text{keep}_j = \lfloor \tau_j \cdot 2^b \rfloor$ from a real keep probability $\tau_j$, together with an alias label $\operatorname{alias}_j \in [L]$.

%Operationally, the implemented branching rule is $J_{\mathrm{out}}=j$ when $\color{red} r \leq k_j$ and $J_{\mathrm{out}}=\operatorname{alias}(j)$ when $\color{red}r>k_j$.
% Thus the realized keep probability is $(k_j+1)/2^b$ rather than an ideal real threshold.
% The same alias table can be combined with either a basic \QROM implementation or a \SelectSwap implementation; only the realization of the two lookup calls changes.
% In the basic \QROM case, the code uses unary iteration to convert the binary address into a one-hot selector, performs bitwise CX loads, and then uncomputes the unary tree.
% In the \SelectSwap case, the address is split into quotient and remainder bits, the quotient selects a block, $\lambda$ data words are loaded into temporary registers, and a controlled swap network routes the selected word to the output register.
% The code chooses a power-of-two block size $\lambda$ by minimizing the surrogate cost $4\lceil L/\lambda \rceil + 4\lambda w$ for a table of $w$-bit words, with $w=b$ for the keep table and $w=n$ for the alias table.
% Finally, note that the present implementation does \emph{not} uncompute the keep, alias, random, or flag registers after the compare-and-swap stage.
% Accordingly, the sampling-based method should be interpreted as preparing the desired address-register marginal rather than a clean ancilla-free pure state.

 Operationally, in the notation of equation~\eqref{eq:alias-classical}, the implemented branching rule outputs $J_{\mathrm{out}}=j$ when $u < \text{keep}_j / 2^b$ and $J_{\mathrm{out}}=\text{alias}_j$ otherwise, where the uniform random variable $u$ is realized coherently by the $b$-qubit register $\ket{\sigma}$ in equation~\eqref{eq:hadamard-register} under the identification $u \leftrightarrow \sigma/2^b$. Consequently, the realized keep probability is the $b$-bit approximation $\text{keep}_j / 2^b$ of the ideal real-valued threshold $\tau_j$. The same alias table can be paired with either a \QROM implementation or a \SelectSwap implementation; only the realization of the two lookup calls differs between the two.

In the \QROM case, the implementation uses unary iteration to convert the binary address into a one-hot selector, performs bitwise CNOT loads, and then uncomputes the unary tree. In the \SelectSwap case, the address is split into quotient and remainder bits: the quotient selects a block, $\lambda$ data words are loaded into temporary registers, and a controlled swap network routes the selected word to the output register. The block size $\lambda$ is chosen as a power of two by minimizing the cost $4\lceil L/\lambda \rceil + 8\lambda w$ %\DL{$4\lceil L/\lambda \rceil + 8\lambda w$?}\DL{different from Guanghao's paper} 
for a table of $w$-bit words, with $w=b$ for the keep table and $w=n$ for the alias table.

We emphasize that the present implementation does \emph{not} uncompute the keep, alias, random, or flag registers after the compare-and-swap stage. Accordingly, the sampling-based routine should be interpreted as preparing the desired marginal on the address register, rather than a clean, ancilla-free pure state.

% \begin{algorithm}[htbp]
% \caption{Alias sampling}
% \label{sec:alias}
% \footnotesize
% \begin{algorithmic}
% \STATE \textbf{Input/Output:} $\alpha=(\alpha_j)_{j=0}^{L-1}\in\mathbb{R}^L$, $L=2^n$, precision $b$, backend in \{basic \QROM, \SelectSwap\}; return $U_{\mathrm{alias}}$ whose address-register marginal approximates the target distribution $p$.
% \STATE Form the probability vector $p$, build the alias table $\{(k_j,\operatorname{alias}(j))\}_{j=0}^{L-1}$, and synthesize \QROM blocks for the keep and alias tables.
% \STATE Allocate the address, random, flag, and \QROM work/output registers, then apply Hadamards to the address and random registers.
% \STATE Load $k_j$, set the comparison flag for the rule $k_j<r$, load $\operatorname{alias}(j)$, and controlled-swap the address register with the alias register when the flag is set.
% \STATE Return the resulting circuit without uncomputing the keep, alias, random, or flag registers.
% \end{algorithmic}
% \end{algorithm}

\subsection{Compilation cost}
We report a common logical proxy
\begin{equation}
    T_{\mathrm{proxy}} = N_T + N_{T^\dagger} + 4N_{\mathrm{CCX}},
\end{equation}
where $N_T$, $N_{T^\dagger}$, and $N_{\mathrm{CCX}}$ denote the numbers of $T$, $T^\dagger$, and Toffoli gates, respectively. This proxy charges each Toffoli as $4T$ under the clean-ancilla model.

For the sampling-based circuits, this already captures the relevant non-Clifford cost, since the synthesized circuits consist solely of Hadamard, Pauli, CNOT, and Toffoli gates, with no arbitrary-angle synthesis stage. We further clarify a few implementation details. The data-lookup oracle exposes a tunable parameter $\lambda$ that controls a trade-off between the number of ancilla qubits and the $T$-gate count~\cite{Low2024tradingtgatesdirty}. In our implementation, we scan over all valid values of $\lambda$ and choose the one that minimizes the $T$ count. Within the quantum lookup table circuit, the only multi-controlled primitive is the Toffoli arising from the unary-iteration tree, and each one-hot data load is realized as a fan-out of singly-controlled NOTs (CNOTs). The multi-controlled gates are decomposed via a Gidney-AND $V$-chain matching the Gleinig-Hoefler asymptotic CNOT bound~\cite{gleinig2021efficient}. In addition, we use a $b$-bit quantum comparator~\cite{gidney2018halving} to determine the swap condition and a cascade of $b$ controlled-SWAP gates to exchange the registers $\text{keep}_j$ and $\text{alias}_j$ accordingly.

%\DL{add reference of comparator}

For the rotation-based circuits, the code uses two levels of costing.
In the broad benchmark script, we record $T_{\mathrm{proxy}}$ together with a synthesis-aware proxy obtained from the explicit single-qubit $R_y$ angles appearing in the logical circuit.
When an explicit Clifford+$T$ circuit is required, the Clifford+$T$ compilation procedure first decomposes controlled rotations to a basic gate set, rewrites
\begin{equation}
    R_y(\theta) = S H R_z(\theta) H S^\dagger,
\end{equation}
and then synthesizes the remaining $R_z$ rotations either exactly at Clifford+$T$ angles or approximately with \texttt{gridsynth} at precision $\epsilon = 2^{-b}$~\cite{ross2016optimal}. Here \texttt{gridsynth} is a single-qubit synthesis routine that takes a target $R_z$ rotation and a tolerance $\epsilon$ and returns an ancilla-free Clifford+$T$ approximation within that tolerance. In other words, it replaces each continuous-angle rotation by a discrete sequence of Clifford and $T$ gates whose length is chosen to meet the requested accuracy.
Thus, for rotation-based methods, the final non-Clifford cost depends both on how many rotations the logical circuit contains and on how hard their angles are to synthesize.
In the dedicated $T$-friendly experiment, the dense template is generated directly from an exactly synthesizable set of single-qubit angles, so no approximation error is introduced on that branch.
We use the same precision parameter $b$ for both families: $b$ controls the alias-table bit width on the sampling-based side and the \texttt{gridsynth} tolerance $\epsilon = 2^{-b}$ on the rotation-based side, so the two families are compared at matched accuracy.
Unless stated otherwise, every figure label that says ``$T$-count'' refers to the compiled Clifford+$T$ $T$-count rather than the logical proxy $T_{\mathrm{proxy}}$.
Likewise, every figure label that says ``total gate count'' refers to the size of the post-compilation circuit used for plotting, rather than the size of the logical template. For rotation-based methods this count is taken after decomposing controlled rotations and synthesizing single-qubit rotations, while for sampling-based methods it is taken after compiling the lookup, comparison, and swap subroutines.

\paragraph{Output metrics}
\label{par:output-metrics}
For rotation-based circuits we evaluate the state fidelity
\begin{equation}
 F_{\mathrm{state}}(\psi,\widetilde{\psi}) = \left| \braket{\psi | \widetilde{\psi}} \right|^2.   
\end{equation}
For sampling-based circuits we instead evaluate the classical fidelity between the target distribution $p$ and the address-register marginal $\widetilde{p}$, 
\begin{equation}
    F_{\mathrm{prob}}(p,\widetilde{p}) = \left( \sum_{j=0}^{L-1} \sqrt{p_j\widetilde{p}_j} \right)^2.
\end{equation}
This distinction is essential here because the rotation-based and sampling-based families produce different types of outputs: the former target the full pure state, while the latter target the correct sampling distribution on the address register.

\section{Theory}
\label{sec:theory}

This section makes the qualitative trade-off between rotation-based and sampling-based preparation precise. Rotation-based methods incur per-rotation Clifford$+T$ synthesis cost but no discretization error; sampling-based methods incur a $b$-bit discretization error but avoid rotation synthesis entirely. We first discuss both costs in closed form for $n=1$, then extend the comparison structurally to general $n$.

\subsection{Single-qubit baseline}
\label{sec:single-qubit-baseline}
We study the single-qubit case as a baseline for the algorithmic comparison developed in \cref{sec:alg}. In this setting, rotation-based methods pay for angle synthesis, whereas sampling-based methods pay for discretization and reversible lookup.
For $n=1$, write
\[
\ket{\psi}=\alpha_0\ket{0}+\alpha_1\ket{1},
\qquad
\alpha_0,\alpha_1\in\mathbb{R},
\qquad
\alpha_0^2+\alpha_1^2=1,
\]
and let $p=(p_0,p_1)$ with $p_j=\alpha_j^2$. For the sampling-based method we compare distributions, while for the rotation-based method we compare the full state. 
% When only the magnitude profile matters, we use
% \[
% \ket{\psi_\theta}=\cos(\theta/2)\ket{0}+\sin(\theta/2)\ket{1},
% \qquad
% \theta\in[0,\pi],
% \]
% since any remaining sign can be absorbed into a Pauli-$Z$ gate.

\paragraph{Rotation-based exact synthesis}
The first question is when a real single-qubit target can be prepared exactly without any approximation step. If exact synthesis is possible, the rotation-based route avoids the compilation overhead associated with arbitrary angles.

\begin{theorem}[Exact single-qubit state synthesis]
\label{thm:single-qubit-real-exact}
Let $\ket{\psi}=\alpha_0\ket{0}+\alpha_1\ket{1}$ with $\alpha_0,\alpha_1\in\mathbb{R}$ and $\alpha_0^2+\alpha_1^2=1$. Then $\ket{\psi}$ is exactly preparable by an ancilla-free Clifford+$T$ circuit if and only if there exists a phase $\omega\in\mathbb{C}$, $|\omega|=1$, such that
\begin{equation}
    \omega \alpha_0,\ \omega \alpha_1 \in \mathbb{Z}[i,1/\sqrt{2}].
\end{equation}
\end{theorem}
This is the state-preparation version of the exact single-qubit Clifford+$T$ synthesis characterization in \cite{kliuchnikov2013exact,giles2014remarks}.
Theorem~\ref{thm:single-qubit-real-exact} shows that exact rotation-based preparation depends on the algebraic form of the target amplitudes, not on an approximation tolerance. If the amplitudes satisfy this criterion, then the state can be prepared exactly by a Clifford+$T$ circuit, so no approximate synthesis step is needed and the state fidelity is $1$.

\paragraph{Rotation-based approximate synthesis}
Most targets do not satisfy the exact criterion, so one must understand the cost of approximation. 
For $\theta\in[0,\pi]$, write $\ket{\psi_\theta}=\cos(\theta/2)\ket{0}+\sin(\theta/2)\ket{1}$. In the single-qubit setting, a real state is prepared by a rotation $R_y(\theta)$. Since $R_y(\theta)=S H R_z(\theta) H S^\dagger$, approximating $R_y(\theta)$ is equivalent, up to fixed Clifford gates, to approximating $R_z(\theta)$. We can therefore use the known Clifford+$T$ synthesis results for $R_z$ rotations.

\begin{theorem}[Approximate single-qubit synthesis]
\label{thm:single-qubit-real-approx}
For every $\theta \in [0, \pi]$ and every $\epsilon > 0$, there exists an ancilla-free Clifford$+T$ circuit $U_\epsilon$ acting on a single qubit such that
$$
    \big\| U_\epsilon \ket{0} - \ket{\psi_\theta} \big\|_2 \;\le\; \epsilon,
$$
where $\|\cdot\|_2$ denotes the $\ell_2$-norm on $\mathbb{C}^2$. Moreover, the $T$-count of $U_\epsilon$ is $O(\log(1/\epsilon))$.
% For every $\theta\in[0,\pi]$ and every $\epsilon>0$, there exists an ancilla-free Clifford+$T$ circuit $U_\epsilon$ such that
% \begin{equation}
%     \|U_\epsilon\ket{0}-\ket{\psi_\theta}\| \le \epsilon.
% \end{equation}
% Moreover, the required $T$-count scales as $O(\log(1/\epsilon))$.
\end{theorem}
This follows from Ross--Selinger synthesis for $R_z$ rotations \cite{ross2016optimal} together with the Clifford conjugation identity $R_y(\theta)=S H R_z(\theta) H S^\dagger$. \cref{thm:single-qubit-real-approx} shows that generic rotation-based single-qubit preparation still has only logarithmic dependence on the target accuracy.

% Combined with \cref{thm:single-qubit-real-exact}, it also explains why ``$T$-friendly'' states are special: the generic $O(\log(1/\epsilon))$ synthesis overhead can collapse to an exact constant cost on a structured subset of targets.

\paragraph{Sampling-based preparation}
% We now turn to the lookup-based implementation from \cref{alg:alias}. As emphasized in \cref{sec:alg}, this branch targets the distribution $p=(p_0,p_1)$ on the address register rather than the signed state itself. The next result quantifies both the logical template cost and the finite-precision discretization error in the single-qubit case.
We next analyze the single-qubit sampling-based method. Unlike the rotation-based method, it does not prepare the signed state $\alpha_0\ket{0}+\alpha_1\ket{1}$ directly; it only aims to reproduce the target probabilities %$p$ 
on the address register. The following theorem gives the proxy $T$-count and the error caused by using a $b$-bit discretization of the alias-sampling threshold.
%\DL{For 1, maybe include other T cost and give a overall cost. The cost also propagate to Cor 3.4 since not all T gate count is included right now.}
\begin{theorem}[Single-qubit sampling cost and accuracy]
\label{thm:single-qubit-lookup}
Let $p=(p_0,p_1)$ with $p_0+p_1=1$, let $b\ge 1$ denote the number of bits used in sampling algorithm in~\cref{sec:alias} to represent $\text{keep}$ value, and let $\widetilde{p}^{(b)}=(\widetilde{p}_0,\widetilde{p}_1):=(\tilde{\alpha}_1^2, \tilde{\alpha}_1^2)$ with
\begin{equation*}
\tilde{\alpha}_0^2=\frac{\text{keep}_0+\sum_{k:\  \text{alias}_k=0}(2^b-\text{keep}_k)}{2^{b+1}},\quad  \tilde{\alpha}_1^2=\frac{\text{keep}_1+\sum_{k:\  \text{alias}_k=1}(2^b-\text{keep}_k)}{2^{b+1}}
\end{equation*}
where $\text{keep}_0$ and $\text{keep}_1$ as well as $\text{alias}_0$ and $\text{alias}_1$ are precomputed given the distribution $p$ and $b$ as in \cite{BabbushGidneyEtAl18}. Then
\begin{enumerate}
    \item The \QROM and \SelectSwap backends have the same proxy $T$-count in this case, namely $T_{\mathrm{proxy}} = 4b + 4$. %\DL{8b-4? }\DL{check this number, we have 1 cswap for $log L$ bits and a comparator of b bit}
    \item Given the $\tau_0,\tau_1$ and $\text{alias}_0,\text{alias}_1$, precomputed for the classical alias sampling method. Let $\text{keep}_j=\lfloor \!2^b \tau_j\rfloor $ for $j\in\{0,1\}$, then $\max_{j \in \{0,1\}}|\widetilde{p}_j-p_j|\le 2^{-b}$.
    
    %then the smaller output probability is $\widetilde{q}^{(b)} = (k+1)/2^{b+1}$.%\DL{Given $\tau_j$, the $keep_j$ can just be $\lfloor 2^b\tau_j \rfloor$ directly}
    % \item Each component satisfies $|\widetilde{p}_j-p_j|\le 2^{-b}$ for $j\in\{0,1\}$.
     \item %The classical fidelity defined in \cref{sec:alg} obeys
     $F_{\mathrm{prob}}(p,\widetilde{p}^{(b)}) \ge (1-2^{-b})^2$.
\end{enumerate}
\end{theorem}
\begin{proof}
For $n=1$, both lookup tables have $L=2$ entries. In this regime the \SelectSwap path reduces to the same two-entry lookup pattern as the \QROM path, so the two backends have identical Toffoli counts. The two-entry lookup itself uses only Pauli-$X$ and CNOT gates, so the only non-Clifford contribution comes from the $b$-bit comparator and the one-bit controlled swap. Under the clean-ancilla proxy of \cref{sec:alg}, this gives $T_{\mathrm{proxy}} = 4b + 4$.

%\DL{starts here}
% Consider the threasholds $\tau_0,\tau_1$ as well as $\text{alias}_0,\text{alias}_1$ precomputed for classical alias sampling method with target probability distribution $p=(\alpha_0^2,\alpha_1^2)$. They satisfy that
% \begin{equation*}
% \alpha_0^2=\frac{\tau_0+\sum_{k:\  \text{alias}_k=0}(1-\tau_k)}{2},\quad  \alpha_1^2=\frac{\tau_1+\sum_{k:\  \text{alias}_k=1}(1-\tau_k)}{2}.
% \end{equation*}
%  Compare the difference between the distribution from quantum algorithm and our target. We can derive that 
% \begin{equation}
% \begin{split}
%      |p_0-\tilde{p}_0|=&| \alpha_0^2-\tilde{\alpha}_0^2|\\=&\left | \frac{(\tau_0-2^{-b}\text{keep}_0)+\sum_{k:\  \text{alias}_k=0}\left(2^{-b}\text{keep}_k-\tau_k\right)}{2} \right|\\
%      \leq&  \frac{|\tau_0-2^{-b}\text{keep}_0|}{2} + \sum_{k:\  \text{alias}_k=0}\frac{|2^{-b}\text{keep}_k-\tau_k|}{2}.
% \end{split}
% \end{equation}
% Use the fact that 
% $\text{keep}_j=\lfloor \!2^b \tau_j\rfloor $ for $j\in\{0,1\}$ and for the case 
% \begin{equation}
%     |\{k| \text{alias}_k=0\}|=1.
% \end{equation}
% We obtain that $|p_0-\tilde{p}_0|\leq 2^{-b}$ and the proof works for $|p_1-\tilde{p}_1|$ as well. Finally,
% \[
% \sum_{j=0}^{1}\sqrt{p_j\widetilde{p}_j}\ge \sum_{j=0}^{1}\min\{p_j,\widetilde{p}_j\}=1-\tfrac12\|p-\widetilde{p}\|_1\ge 1-2^{-b},
% \]
% and squaring gives $F_{\mathrm{prob}}(p,\widetilde{p}^{(b)}) \ge (1-2^{-b})^2$.

Let $\tau_0, \tau_1$ and $\text{alias}_0, \text{alias}_1$ denote the thresholds and alias labels precomputed by the classical alias sampling method for the target probability distribution $p = (\alpha_0^2, \alpha_1^2)$. By construction, they satisfy
\begin{equation*}
\alpha_0^2 = \frac{\tau_0 + \sum_{k:\ \text{alias}_k = 0}(1 - \tau_k)}{2}, \qquad \alpha_1^2 = \frac{\tau_1 + \sum_{k:\ \text{alias}_k = 1}(1 - \tau_k)}{2}.
\end{equation*}
We now bound the deviation between the distribution $\widetilde{p}$ produced by the quantum algorithm and the target $p$. A direct calculation gives
\begin{equation}
\begin{split}
    |p_0 - \widetilde{p}_0| &= |\alpha_0^2 - \tilde{\alpha}_0^2| \\
    &= \left| \frac{(\tau_0 - 2^{-b}\,\text{keep}_0) + \sum_{k:\ \text{alias}_k = 0}\left(2^{-b}\,\text{keep}_k - \tau_k\right)}{2} \right| \\
    &\leq \frac{|\tau_0 - 2^{-b}\,\text{keep}_0|}{2} + \sum_{k:\ \text{alias}_k = 0} \frac{|2^{-b}\,\text{keep}_k - \tau_k|}{2}.
\end{split}
\end{equation}
Using $\text{keep}_j = \lfloor 2^b \tau_j \rfloor$ for $j \in \{0, 1\}$, together with the fact
\begin{equation}
    |\{k \mid \text{alias}_k = 0\}| = 1,
\end{equation}
we obtain $|p_0 - \widetilde{p}_0| \leq 2^{-b}$, and an similar argument yields $|p_1 - \widetilde{p}_1| \leq 2^{-b}$. Consequently,
\[
\sum_{j=0}^{1}\sqrt{p_j \widetilde{p}_j} \;\ge\; \sum_{j=0}^{1}\min\{p_j, \widetilde{p}_j\} \;=\; 1 - \tfrac{1}{2}\|p - \widetilde{p}\|_1 \;\ge\; 1 - 2^{-b},
\]
and squaring both sides gives $F_{\mathrm{prob}}(p, \widetilde{p}^{(b)}) \ge (1 - 2^{-b})^2$.

% Now let $q=\min\{p_0,p_1\}$ and $x=2q$. The alias table stores $k=\operatorname{round}\!\left(x(2^b-1)\right)$, and the circuit keeps the smaller bin when the $b$-bit random register satisfies $r\le k$. Since the outer address choice is uniform over two bins, the smaller output probability is $\widetilde{q}^{(b)}=(k+1)/2^{b+1}$. Also,
% \[
% \left|\widetilde{q}^{(b)}-q\right|
% =
% \frac{|k+1-2^b x|}{2^{b+1}}
% \le
% \frac{|k-x(2^b-1)|+|1-x|}{2^{b+1}}
% \le 2^{-b},
% \]
% so both components satisfy $|\widetilde{p}_j-p_j|\le 2^{-b}$. 

\end{proof}

%\DL{change Cor 3.4 after total T formula revised}
%\DL{changes in T proxy also propagates into here.}
\begin{corollary}
\label{cor:single-qubit-lookup-threshold}
%\DL{check}
Let $\eta\in(0,1)$ and choose $b=\left\lceil \log_2 \frac{2}{\eta} \right\rceil$. Then the single-qubit sampling-based implementation satisfies $F_{\mathrm{prob}}(p,\widetilde{p}^{(b)}) \ge 1-\eta$ with $T_{\mathrm{proxy}} \le 4\left\lceil \log_2 \frac{2}{\eta} \right\rceil + 4$.
\end{corollary}
\begin{proof}
By \cref{thm:single-qubit-lookup}, this choice gives $T_{\mathrm{proxy}}=4b+4$ and $F_{\mathrm{prob}}(p,\widetilde{p}^{(b)}) \ge (1-2^{-b})^2$. Since $2^{-b}\le \eta/2$, we have $(1-2^{-b})^2 \ge (1-\eta/2)^2 \ge 1-\eta$.
\end{proof}

\cref{cor:single-qubit-lookup-threshold} shows that the $T$ count of sampling-based method also has logarithmic dependence on the requested accuracy, but the mechanism is different from the rotation case. The cost depends only on the discretization parameter $b$; it is insensitive to the target state.
%
% A second implication of \cref{thm:single-qubit-lookup} is specific to the present implementation: because $\widetilde{q}^{(b)}\ge 2^{-(b+1)}$, finite-precision lookup preparation never produces an exact zero on the smaller bin. Thus, even a computational-basis state is only approximated at finite $b$. This is not a limitation of alias sampling in principle; it is a consequence of the current rounding and comparison convention.
% This theorem also reveals an artifact of the current implementation. Since
% \[
% \widetilde q^{(b)}=\frac{k+1}{2^{b+1}} \ge 2^{-(b+1)},
% \]
% the smaller output probability is never exactly zero at finite precision. Therefore even a computational-basis distribution, such as $(1,0)$, is only approximated. This is caused by the present rounding and comparison rule, not by alias sampling itself.\DL{Check this point}

\subsection{Structural comparison at fixed precision}
\label{sec:structural-comparison}

At fixed precision $b$, the sampling-based preparation is confined to a finite probability grid. By \cref{thm:single-qubit-lookup}, the amplitudes produced by the sampling method lie in the finite set $\{k/2^{b+1} : k = 0, 1, \ldots, 2^{b+1}\}$, so only finitely many binary marginals are realized exactly at precision $b$. As $b$ increases, this grid becomes arbitrarily fine, and the union over all $b$ is dense in the space of binary distributions---but at any \emph{fixed} $b$, the set of reachable distributions remains finite.

Exact Clifford+$T$ preparation, by contrast, exploits a different arithmetic structure. \cref{thm:single-qubit-real-approx} implies that the exactly Clifford+$T$-preparable real single-qubit states are dense in $\{\ket{\psi_\theta} : \theta \in [0,\pi]\}$. Applying the continuous map $$\ket{\psi_\theta} \mapsto (\cos^2(\theta/2), \sin^2(\theta/2))$$ shows that the corresponding exactly preparable probability vectors are dense in the set of all binary distributions. Since a dense subset of an infinite space cannot be contained in any finite set, some exactly Clifford+$T$-preparable states necessarily lie outside the fixed-$b$ lookup grid.

This contrast highlights the basic structural difference between two approaches: fixed-precision lookup is governed by a discretization grid, whereas rotation-based synthesis leverages exact arithmetic structure. This is the mechanism underlying the $T$-friendly advantage observed later in \cref{sec:resource}.

\section{Numerical experiments}
\label{sec:numerical}

In this section, we compare the two algorithmic families---dense and sparse rotation-based preparation versus sampling-based alias sampling with two backends---across several benchmark states, and we report the resulting logical resource estimates.

\subsection{Structured and \texorpdfstring{$T$}{T}-friendly benchmarks}
\label{sec:resource}

\paragraph{W and Dicke states}
We first benchmark three structured state families. The $W$ state is a canonical multipartite entangled state whose entanglement is notably robust to the loss of one qubit~\cite{dur2000three}. Dicke states generalize this fixed-excitation structure; they arise naturally in collective-emission physics~\cite{dicke1954coherence} and also appear in multiparty quantum networking protocols~\cite{prevedel2009experimental}. For $n$ qubits, the $W$ state is
\[
\ket{W_n}=\frac{1}{\sqrt{n}} \sum_{j=1}^{n} \ket{0^{j-1}10^{n-j}},
\]
and the Dicke state of Hamming weight $k$ is
\[
\ket{D_{n}^{k}}=\binom{n}{k}^{-1/2}\sum_{|x|=k}\ket{x}.
\]
We consider $\ket{W_n}$ together with $\ket{D_{n}^{2}}$ and $\ket{D_{n}^{3}}$ at $n=8$, and compare the dense and sparse rotation-based methods with \QROM alias sampling and \SelectSwap alias sampling. \Cref{fig:w-dicke-inf-vs-t} shows the resulting tradeoff. For these resource states, the qualitative behavior is largely controlled by support size. The $W_8$ state has support $8$, so the sparse rotation method can exploit this extreme sparsity and gives the best tradeoff over the plotted range, whereas the sampling-based methods pay overhead that is not recovered at this scale. The Dicke states are less sparse, with $|\mathrm{supp}(D_{8}^2)|=\binom{8}{2}=28$ and $|\mathrm{supp}(D_{8}^3)|=\binom{8}{3}=56$. As the support grows, the advantage of sparse rotation weakens, and in the low-infidelity regime the sampling-based methods become more favorable. In all three panels, \SelectSwap uniformly improves on \QROM.
% \DL{what is the tradeoff? we need explicitly interpret the numeric. We can mention that the W state is too sparse and in that case sampling based method is less competitive. Also we can point out that once Dick state is becoming more and more dense, the dense method is becoming denser and in that case dense method naturally performs better and when it become dense the sampling method perform better.}

\begin{figure}[htbp]
    \centering
    \includegraphics[width=\textwidth]{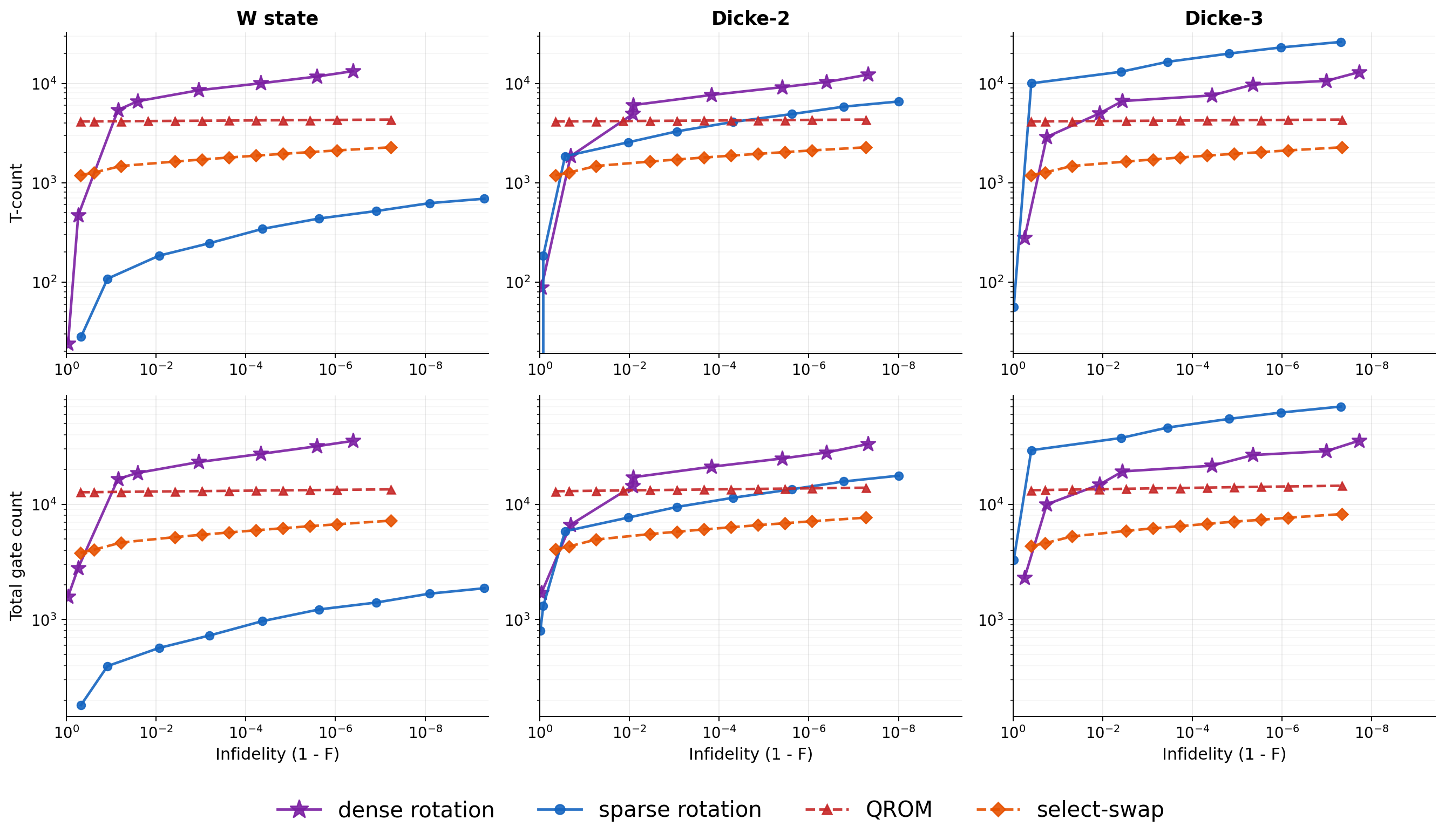}
    \caption{Infidelity versus $T$-count and total gate count for the $8$-qubit $W$, Dicke-$2$, and Dicke-$3$ states ($\ket{W_8}$, $\ket{D_8^2}$, and $\ket{D_8^3}$). The top row reports $T$-count and the bottom row reports total gate count. The $W$ state favors sparse rotation over the plotted range, while Dicke-$2$ and Dicke-$3$ favor sampling-based methods in the low-infidelity regime. In all columns, \SelectSwap improves on \QROM.}
    \label{fig:w-dicke-inf-vs-t}
\end{figure}

As an implementation-level complement to the logical resource counts, \Cref{fig:w-dicke-runtime} reports the measured total synthesis time for the same benchmark family.

\begin{figure}[htbp]
    \centering
    \includegraphics[width=\textwidth]{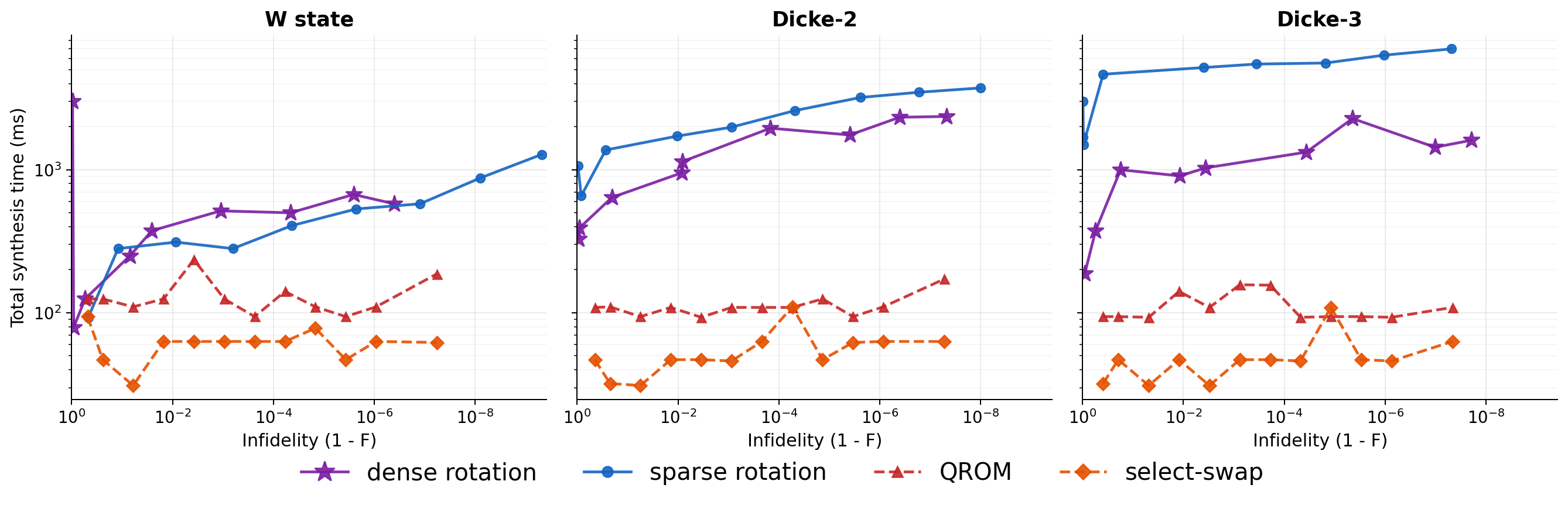}
    \caption{Infidelity versus total synthesis time, in milliseconds, for the $8$-qubit $W$, Dicke-$2$, and Dicke-$3$ states ($\ket{W_8}$, $\ket{D_8^2}$, and $\ket{D_8^3}$). Each panel shows the raw compiled points for dense rotation, sparse rotation, \QROM, and \SelectSwap.}
    \label{fig:w-dicke-runtime}
\end{figure}

\paragraph{Dense and sparse}
To complement the structured examples above, we also test three synthetic benchmark families at $n=8$: a dense random state with support size $2^{n-1}$, a sparse uniform state with $n$ equal-magnitude nonzero amplitudes, and a sparse random state with $n$ random nonzero amplitudes. \Cref{fig:dense-sparse-inf-vs-t} shows the resulting tradeoff.

\begin{figure}[htbp]
    \centering
    \includegraphics[width=\textwidth]{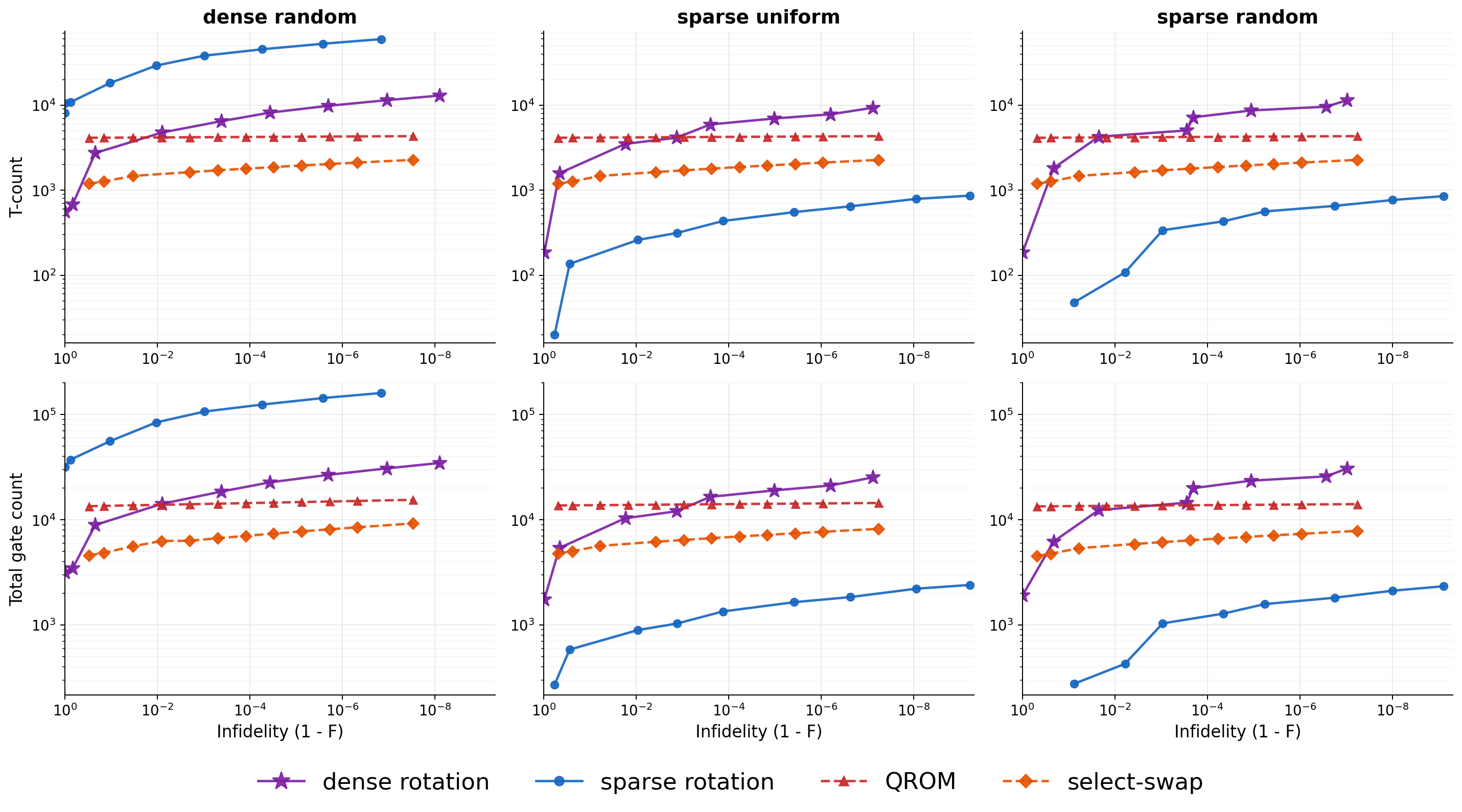}
    \caption{Infidelity versus $T$-count and total gate count for $8$-qubit dense random, sparse uniform, and sparse random states. The top row reports $T$-count and the bottom row reports total gate count. Sparse rotation gives the best tradeoff on the two sparse families, while the dense random instance favors \SelectSwap over the plotted range. \SelectSwap is uniformly better than \QROM.}
    \label{fig:dense-sparse-inf-vs-t}
\end{figure}

\paragraph{\texorpdfstring{$T$}{T}-friendly states}
We next consider a structured benchmark designed to isolate a regime in which rotation-based preparation is especially favorable. By \emph{$T$-friendly} we mean that every single-qubit rotation appearing in the dense circuit is exactly synthesizable in Clifford+$T$. For each $n$, we generate five random states by instantiating the dense template with angles chosen from an exactly synthesizable library. These states are therefore not generic random inputs; they are meant to probe a best-case regime in which the dense rotation-based method can fully exploit arithmetic structure in the target state. The resulting dense circuit is exact at the logical level, so its $T$-count is determined entirely by the chosen angle library and its state infidelity is zero. We compare this dense $T$-friendly circuit against sparse rotation compiled with \texttt{gridsynth}, \QROM alias sampling, and \SelectSwap alias sampling. For the compiled methods we sweep the precision parameter $b$ up to $10$ and evaluate the output using the fidelity metrics from \cref{sec:alg}.
As shown in \Cref{fig:t-friendly-inf-vs-t}, this yields a clear best-case separation in favor of the dense rotation-based construction.

\begin{figure}[htbp]
    \centering
    \includegraphics[width=\textwidth]{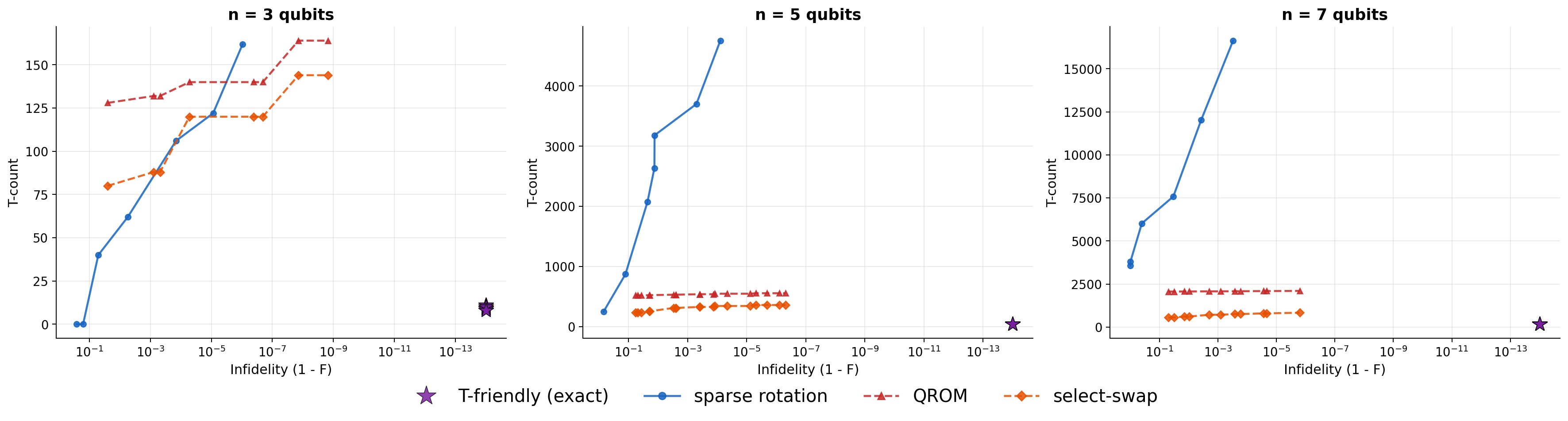}
    \caption{Infidelity versus $T$-count for the $T$-friendly benchmark. Each panel fixes the number of qubits $n$. The dense $T$-friendly method gives the best overall tradeoff across almost all tested instances, with only small-system, low-precision cases remaining competitive for sparse rotation. \SelectSwap is uniformly better than \QROM, and once $n \ge 4$ the dense $T$-friendly construction has the lowest $T$-count throughout the tested precision range.}
    \label{fig:t-friendly-inf-vs-t}
\end{figure}

\subsection{Tensor
Hypercontraction for FeMoCo and $\text{H}_2\text{O}$}
\label{sec:THC}
Quantum state preparation has been an important subroutine for constructing block input models, aiming at loading necessary classical information to initialize quantum simulation. For quantum simulation of a many-body system, the state to be prepared is determined by the Hamiltonian of the physical system of interest. 

In this section, we present two concrete examples of quantum simulation: the FeMo cofactor of nitrogenase, commonly known as FeMoCo, which plays a central role in biological nitrogen fixation, and water, a fundamental molecular system of broad chemical and physical importance. To describe the electronic structure problem, we work in a second-quantized formulation built from a set of spatial orbitals and explicit spin labels. In this representation, the Hamiltonian is decomposed into a one-electron contribution and an interaction contribution,
\begin{align}
H &= T + V, \\
T &= \sum_{\sigma \in \{\uparrow,\downarrow\}} \sum_{p,q=1}^{N_{\mathrm{orb}}/2}
T_{pq}\, a^\dagger_{p,\sigma} a_{q,\sigma}, \\
V &= \frac{1}{2}
\sum_{\alpha,\beta \in \{\uparrow,\downarrow\}}
\sum_{p,q,r,s=1}^{N_{\mathrm{orb}}/2}
V_{pqrs}\, a^\dagger_{p,\alpha} a_{q,\alpha} a^\dagger_{r,\beta} a_{s,\beta}.
\end{align}
The integer $N_{\mathrm{orb}}$ is the total number of spin-orbital basis functions, implying that there are $n_{\text{orb}}:=N_{\mathrm{orb}}/2$ spatial orbitals indexed by $p,q,r,s$. Spin is indicated by $\alpha,\beta,\sigma \in \{\uparrow,\downarrow\}$. As usual, $a^\dagger_{p,\sigma}$ creates and $a_{p,\sigma}$ annihilates a fermion in spatial orbital $p$ with spin $\sigma$. The matrix $T_{pq}$ collects the effective single-particle terms, while $V_{pqrs}$ contains the two-particle Coulomb potentials. We retain the prefactor $1/2$ explicitly in the definition of the interaction operator. For detailed introduction, we refer to~\cite{lin2019mathematical,ashcroft1976solid}. %\SZ{citation?}\DL{added 2} 

Assuming a real spatial basis $\{\phi_i\}_{i=1}^{N_{\mathrm{orb}}/2}$, the four-index Coulomb integral is
\begin{equation}
V_{pqrs}
=
\iint d\mathbf{r}_1\, d\mathbf{r}_2\,
\frac{
\phi_p(\mathbf{r}_1)\phi_q(\mathbf{r}_1)\phi_r(\mathbf{r}_2)\phi_s(\mathbf{r}_2)
}{
|\mathbf{r}_1-\mathbf{r}_2|
}.
\end{equation}
The one-body coefficient appearing in $T$ is then
\begin{equation}
T_{pq} = h_{pq} - \frac{1}{2}\sum_{r=1}^{N_{\mathrm{orb}}/2} V_{prrq},
\end{equation}
where the bare one-electron integral takes the form
\begin{equation}
h_{pq}
=
\int d\mathbf{r}\,
\phi_p(\mathbf{r})
\left[
-\frac{1}{2}\nabla^2
-
\sum_A \frac{Z_A}{|\mathbf{r}-\mathbf{R}_A|}
\right]
\phi_q(\mathbf{r}).
\end{equation}
Here, the summation over $A$ runs over all nuclei, with $Z_A$ and $\mathbf{R}_A$ denoting the corresponding nuclear charges and positions.

For many electronic structure problems, the two electron integral $V_{pqrs}$ can be factorized in the form that
\begin{equation}
    V_{pqrs} \approx \sum_{\mu,\nu}^{M} \chi_{p}^{\mu} \chi_{q}^{\mu} \xi_{\mu \nu} \chi_r^{\nu} \chi_{s}^{\nu}
\end{equation}
and the approximated two-body interaction can be written as
\begin{equation}
   \frac{1}{2}
\sum_{\alpha,\beta \in \{\uparrow,\downarrow\}}
\sum_{p,q,r,s=1}^{N_{\mathrm{orb}}/2}
\sum_{\mu,\nu}^{M} \chi_{p}^{\mu} \chi_{q}^{\mu} \xi_{\mu \nu} \chi_r^{\nu} \chi_{s}^{\nu}\, a^\dagger_{p,\alpha} a_{q,\alpha} a^\dagger_{r,\beta} a_{s,\beta}.
\end{equation}
Define a basis rotation
\begin{equation}
    c^{\dagger}_{\mu,\sigma} = \sum_{p=1}^{N_{\mathrm{orb}}/2}\chi_{p}^{\mu} a_{p,\sigma}^{\dagger}, \quad c_{\mu,\sigma}=\sum_{p=1}^{N_{\mathrm{orb}}/2}\chi_{p}^{\mu} a_{p,\sigma}.
\end{equation}
With this basis, the Hamiltonian can be rewritten in tensor hypercontraction form~\cite{rubin2023fault,lee2021even},
\begin{equation}
    H_{\text{THC}}= \sum_{\sigma \in \{\uparrow,\downarrow\}} \sum_{l=1}^{N_{\mathrm{orb}}/2} t_l n_{l,\sigma}+\frac{1}{2} \sum_{\alpha,\beta \in \{\uparrow,\downarrow\}} \sum_{\mu,\nu}^{M} \xi_{\mu\nu} n_{\mu,\alpha} n_{\nu,\beta}.
\end{equation}
Quantum simulation of the system requires preparation of the state
\begin{equation}
    \label{eq:thc-coeff-state}
    \frac{1}{\sqrt{\lambda_{\mathrm{THC}}}}\left [\sum_{l=0}^{N_{\mathrm{orb}}/2-1} \sqrt{|t_l|}\ket{l}\ket{M}+\frac{1}{\sqrt{2}}\sum_{\mu,\nu=0}^{M-1}\sqrt{|\xi_{\mu,\nu}|} \ket{\mu} \ket{\nu} \right ].
\end{equation}
where
\begin{equation}
    \lambda_{\mathrm{THC}} = \sum_{l=0}^{N_{\mathrm{orb}}/2-1} |t_l| + \frac{1}{2}\sum_{\mu,\nu=0}^{M-1} |\xi_{\mu,\nu}|.
\end{equation}
%\DL{maybe use $\lambda_{\text{THC}}$ or $\alpha_{\text{THC}}$ as more commonly used notation for subnormalization factor. Add the full plot for water and FeMoCo.}

% We consider two nonorthogonal tensor hypercontraction Hamiltonians: the FeMoCo Hamiltonian computed in \cite{PRXQuantum.2.030305} with $M=350, n_{\mathrm{orb}}=76$ and chemical accuracy, and a water Hamiltonian with $M=258$ and $n_{\mathrm{orb}}=24$. \textcolor{blue}{For water molecule, we obtain the THC data using the open software \texttt{CoQuí}\cite{coqui_github}, which implements the efficient algorithm developed in \cite{yeh2023low, yeh2024low}. }
We consider two nonorthogonal tensor hypercontraction (THC) Hamiltonians: the FeMoCo Hamiltonian reported in \cite{PRXQuantum.2.030305}, with $M=350$ and $n_{\mathrm{orb}}=76$ at chemical accuracy, and a water Hamiltonian with $M=258$ and $n_{\mathrm{orb}}=24$.
For the water molecule, we generate the THC data using the open-source software \texttt{CoQuí}\cite{coqui_github}, which implements the efficient algorithms developed in \cite{yeh2023low, yeh2024low}.
For these large instances, we report the dense rotation method as well as sampling-based methods \QROM and \SelectSwap. We do not report the sparse rotation method here because, in our current implementation, compiling and evaluating them at this scale was prohibitively time-consuming. \Cref{fig:tworow} shows the resulting infidelity tradeoff as a function of logical $T$-count and total compiled gate count.

%DL:I don't know why sparse method is good in this case. Will double check code.

\begin{figure}[htbp]
    \centering
    \includegraphics[width=0.9\textwidth]{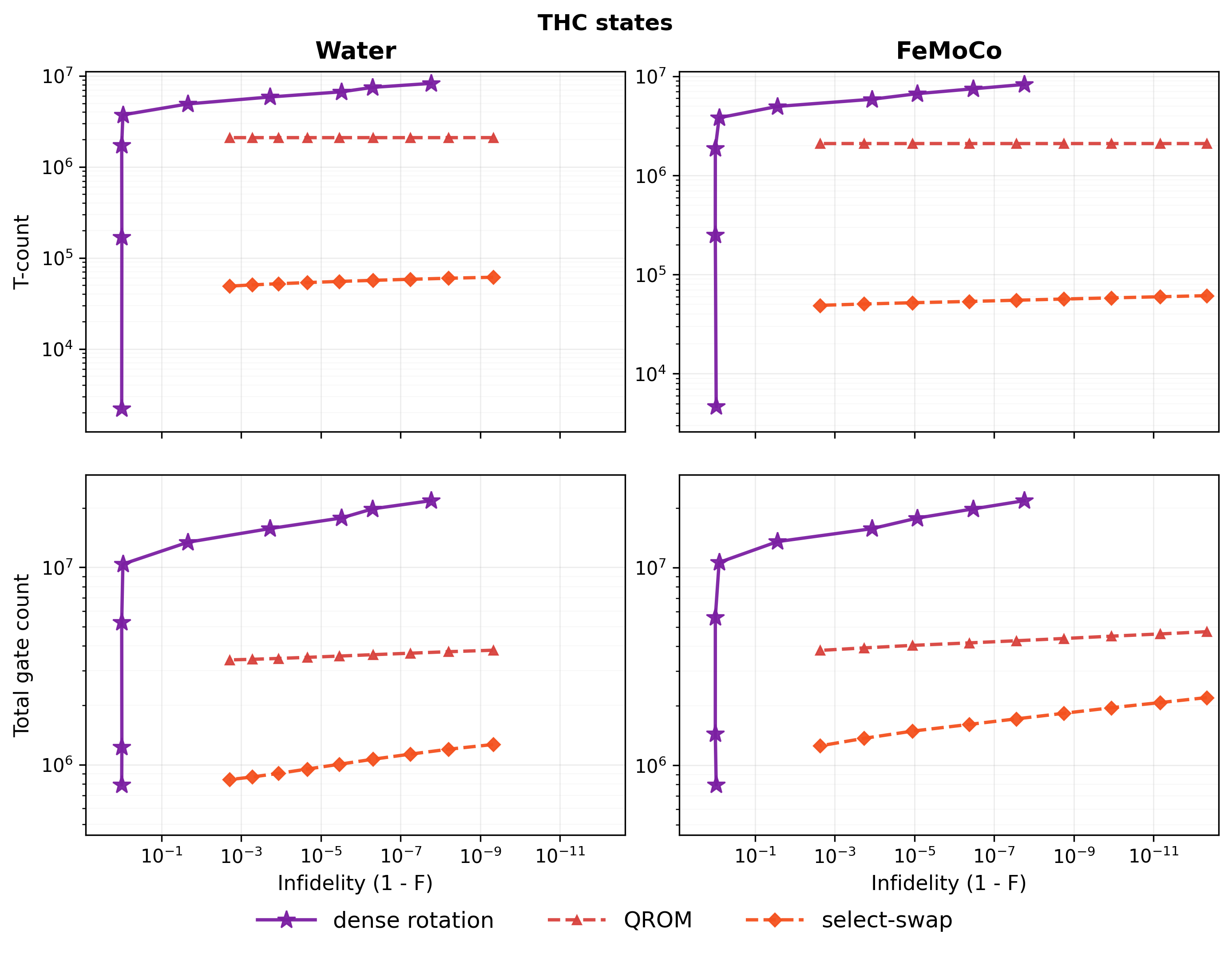}
    \caption{Comparison of state-preparation cost versus infidelity for the desne rotation method and two sampling-based methods including \QROM and \SelectSwap  for the FeMoCo electronic structure Hamiltonian with $M=350, n_{orb}=76$ and $n=17$ and $\text{H}_2\text{O}$ electronic structure Hamiltonian with $M=258, n_{orb}=24$ and $n=17$. The top panels reports logical $T$-count, and the bottom panels reports total compiled gate count. We omit the sparse rotation-based method because, in our current implementation, compiling and evaluating them at this scale was prohibitively time-consuming.}
    \label{fig:tworow}
\end{figure}
To obtain a smaller THC-like benchmark with the same coefficient-state structure as \eqref{eq:thc-coeff-state}, we also generate a synthetic small-scale THC-like instance at $n=8$. For this toy instance, we use a $16\times16$ pair-index grid, so $\lceil \log_2(16^2)\rceil = 8$, and set $n_{\mathrm{orb}}=16$ and $M=15$. We sample real coefficients
\[
    t_l \sim \mathcal{N}(0,1), \qquad
    \chi_p^\mu \sim \mathcal{N}(0,1), \qquad
    \xi_{\mu\nu}=\xi_{\nu\mu}, \quad \xi_{\mu\nu} \sim \mathcal{N}(0,1),
\]
normalize the columns of $\chi$ according to the paper convention, absorb the resulting column norms into $\xi$, and then instantiate \eqref{eq:thc-coeff-state} with these sampled coefficients. The resulting real-amplitude state is encoded with the same pair-index map used for the FeMoCo benchmark. \Cref{fig:thc-toy} shows the corresponding tradeoff for this smaller instance.

\begin{figure}[htbp]
    \centering
    \includegraphics[width=\textwidth]{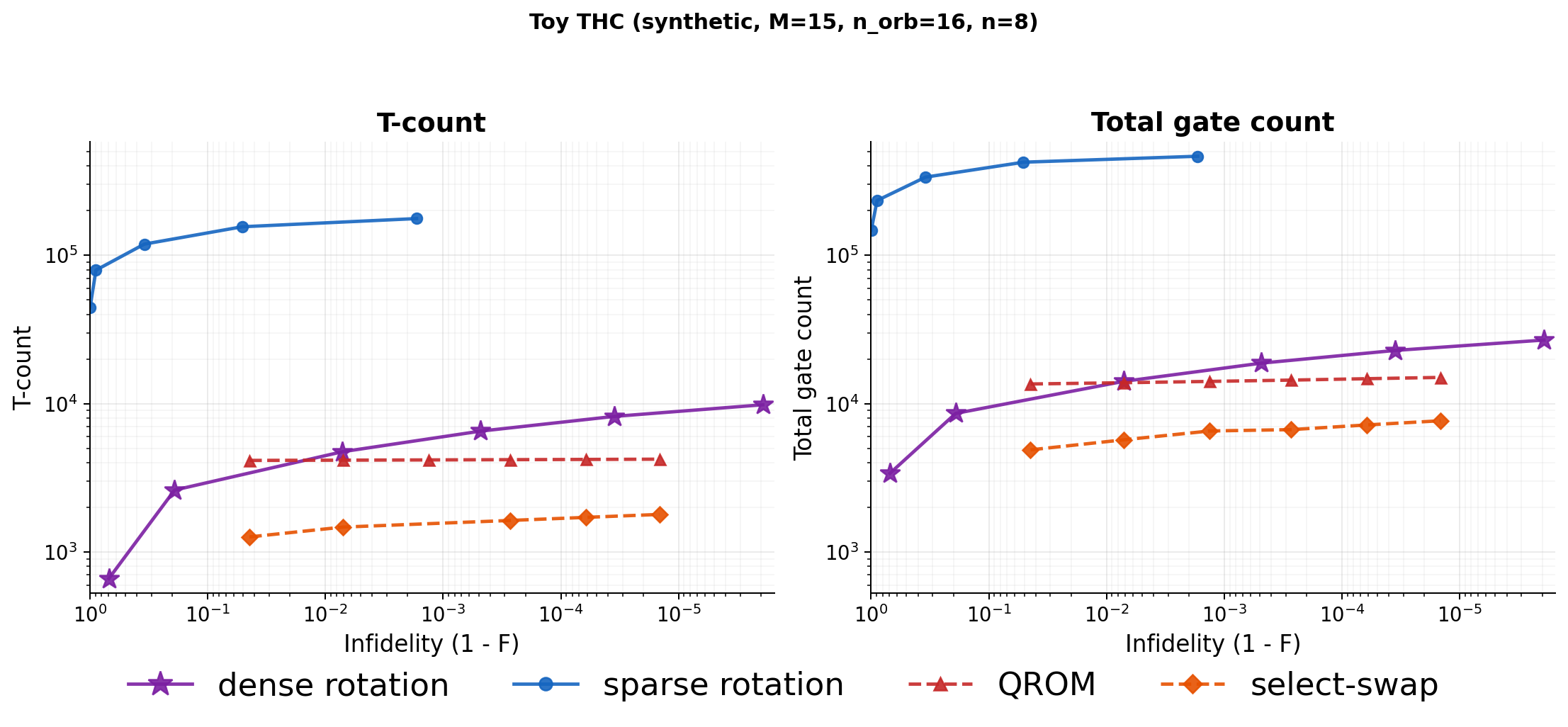}
    \caption{Infidelity versus $T$-count and total gate count for the synthetic toy THC benchmark at $n=8$ with $M=15$ and $n_{\mathrm{orb}}=16$. The toy instance preserves the coefficient-state structure of the FeMoCo benchmark while replacing the chemistry-derived coefficients with real Gaussian samples.}
    \label{fig:thc-toy}
\end{figure}
\subsection{Sachdev–Ye–Kitaev model}
\label{sec:syk}
The Sachdev--Ye--Kitaev (SYK) model is a paradigmatic strongly interacting fermionic system with random, all-to-all few-body couplings that has emerged as a useful meeting point of condensed matter physics and quantum field theory. 

From the condensed-matter perspective, it provides a solvable yet highly nontrivial model of non-Fermi-liquid behavior, capturing the absence of well-defined quasiparticles, strong scattering, and aspects of strange-metal phenomenology in an idealized setting. At the same time, from the quantum-field-theoretic perspective, SYK may be viewed as a $0+1$-dimensional large-$N_{\mathrm M}$ interacting field theory whose correlation functions and low-energy dynamics admit a controlled analysis through Schwinger--Dyson equations, emergent conformal symmetry, and related effective descriptions. Owing to this combination of physical relevance and analytical tractability, the SYK model has become a standard benchmark for studying strongly correlated quantum matter, many-body chaos, and the capabilities of classical and quantum algorithms for simulating interacting systems.
In the Majorana form used here, the Hamiltonian is
\begin{equation}
H_{\mathrm{SYK}}
=
    \frac{1}{\sqrt{\binom{N_{\mathrm M}}{4}}}
    \sum_{1 \le a < b < c < d \le N_{\mathrm M}}
J_{abcd}\,\gamma_a\gamma_b\gamma_c\gamma_d,
\qquad
J_{abcd}\sim \mathcal{N}(0,1),
\end{equation}
where $\gamma_1,\ldots,\gamma_{N_{\mathrm M}}$ are Majorana operators. In the qubit encoding used for our benchmark data, $N_{\mathrm M}=2n$, so the instances plotted at $n=5,8,10$ correspond to $N_{\mathrm M}=10,16,20$ Majorana modes. For the present real-amplitude benchmark, \Cref{fig:syk-inf-vs-t} shows the infidelity--$T$-count tradeoff for real-surrogate SYK ground states at $n=5,8,10$.

For each instance, we diagonalize the SYK Hamiltonian and take its ground-state eigenvector $\ket{\psi_{\mathrm{gs}}}$. Then we form the benchmark state by taking the entrywise real part and renormalizing:
\begin{equation}
    \ket{\psi_{\mathrm{real}}}
    =
    \frac{\Re(\ket{\psi_{\mathrm{gs}}})}{\|\Re(\ket{\psi_{\mathrm{gs}}})\|_2}.
\end{equation}
Thus, the prepared state is not the exact SYK ground state, but a normalized real-amplitude surrogate of it.
\begin{figure}[htbp]
    \centering
    \includegraphics[width=\textwidth]{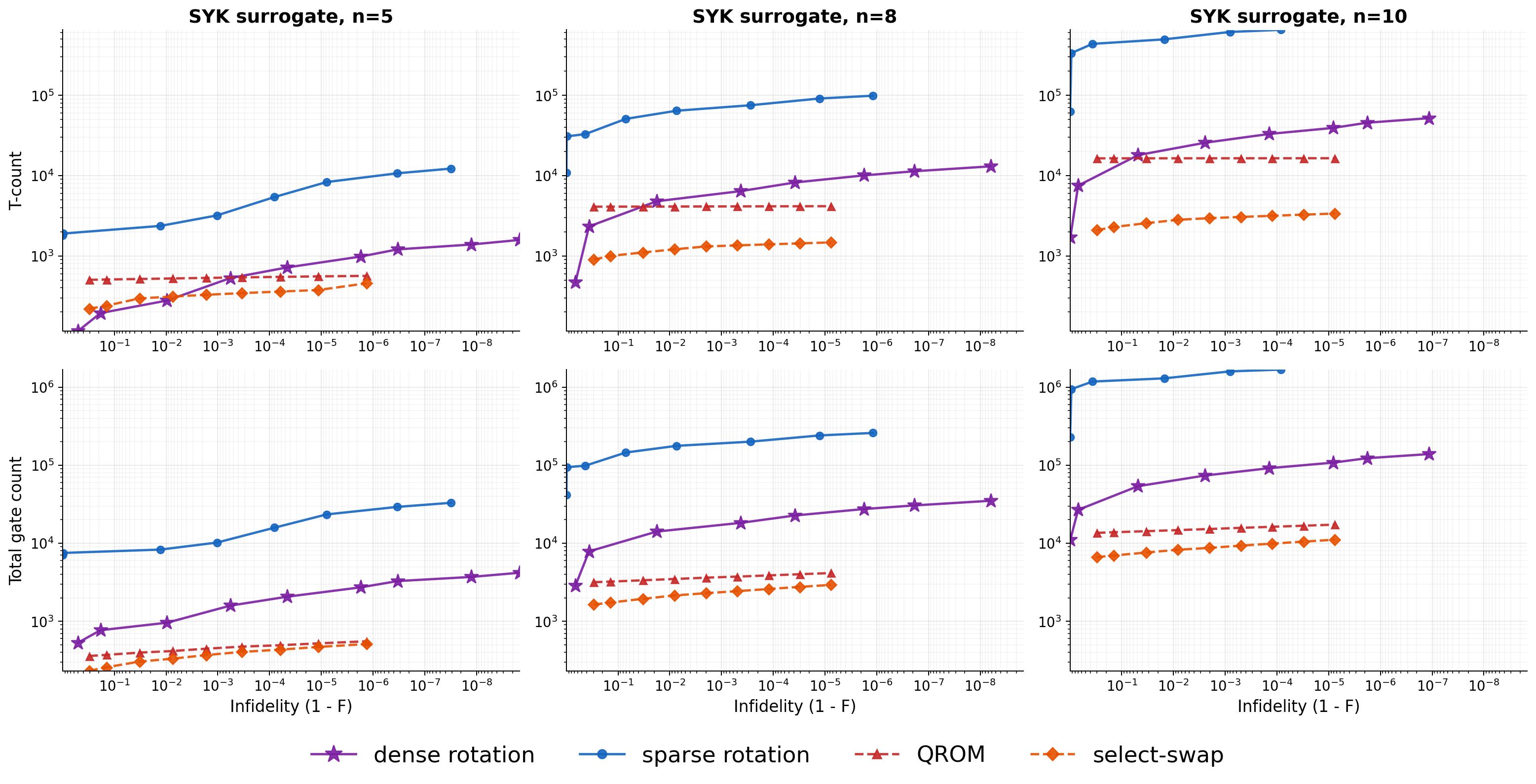}
    \caption{Infidelity versus $T$-count and total gate count for SYK real-surrogate ground states at $n=5,8,10$. The top row reports $T$-count and the bottom row reports total gate count.}
    \label{fig:syk-inf-vs-t}
\end{figure}

\subsection{Magnus expansion}
Simulating time-dependent Hamiltonians is a central task in quantum computation, driven by physical processes like external controls and algorithmic designs such as adiabatic quantum computing. Furthermore, techniques like the interaction picture often convert static problems into time-dependent forms to achieve greater algorithmic efficiency. Unlike time-independent cases, simulating these dynamics presents a unique challenge: developing efficient quantum algorithms capable of evaluating the time-ordered evolution operator. 

We next review high-order Magnus expansion~\cite{blanes2009magnus,fang2025time,fang2025high,an2022time} and introduce the quantum state preparation required for the method. Consider the linear evolution equation
\begin{equation}
\dot{U}(t)=A(t)U(t), \qquad U(0)=U_0,
\end{equation}
which reduces to the time-dependent Schr\"odinger equation when
\begin{equation}
A(t)=-iH(t),
\end{equation}
where $H(t)$ is Hermitian. Its formal solution is given by the time-ordered exponential
\begin{equation}
U(t)=\mathcal{T}\exp\!\left(\int_0^t A(s)\,ds\right)U_0.
\end{equation}
% The Magnus expansion rewrites the propagator as a single exponential,
% \begin{equation}
% U(t)=\exp(\Omega(t))\,U_0,
% \end{equation}
% where
% \begin{equation}
% \Omega(t)=\sum_{n=1}^{\infty}\Omega_n(t),
% \end{equation}
% with
% \begin{equation}
% \Omega_1(t)=\int_0^t A(s)\, \mathrm{d}s,
% \end{equation}
% and for $n\geq 2$,
% \begin{equation}
%     \Omega_n(t)=\sum_{j=1}^{n-1} \frac{B_j}{j!} \sum_{\substack{k_1+k_2\dots+k_j=n-1 \\k_1 \geq1, \dots k_j \geq 1}} \int_{0}^{t} \,
%        \text{ad}_{\Omega_{k_j}(s)} \, \cdots \, \text{ad}_{\Omega_{k_2}(s)} 
%           \, \text{ad}_{\Omega_{k_1}(s)} A(s) \, \mathrm{d}s,
% \end{equation}
% where $B_j$ are the Bernoulli numbers with $B_1=-\frac{1}{2}$ and $\mathrm{ad}_X(Y)=[X,Y]$.
The $p$-th order Magnus expansion can then be defined as
\begin{equation*}
    \Omega_{(p)}(t,0)=\sum_{k=1}^{p}\sum_{\pi\in S_k} C_{\pi,k} \int_{0}^{t} dt_1 \int_{0}^{t_1} dt_2 \cdots \int_{0}^{t_{k-1}} dt_{k} A(t_{\pi(1)})A(t_{\pi(2)})\cdots A(t_{\pi(k)}),
\end{equation*}
with $
    C _{\pi,k}= \frac{(-1)^{d_a(\pi)}}{k} \times \frac{1}{ {k-1 \choose d_a(\pi)}}$
where $d_a(\pi)=\left|\{\,i\in\{1,\dots,k-1\}|\pi(i)>\pi(i+1)\}\right|$ denotes the number of descents of permutation $\pi$, $k$ denotes the length of permutation and $S_k$ denotes all such permutation of length $k$.

A straightforward implementation for the high-order Magnus expansion requires preparation of a sequence of states
\begin{equation}
    \frac{1}{\sqrt{k!}} \sum_{\pi \in S_k} C_{\pi,k} \ket{\pi(0)} \ket{\pi(1)}\dots \ket{\pi(k-1)}.
\end{equation}
For the benchmark states with $k=3$ and $k=4$ generated from this construction, \Cref{fig:magnus-inf-vs-counts} shows the resulting infidelity tradeoff as a function of $T$-count and total gate count.

\begin{figure}[htbp]
    \centering
    \includegraphics[width=\textwidth]{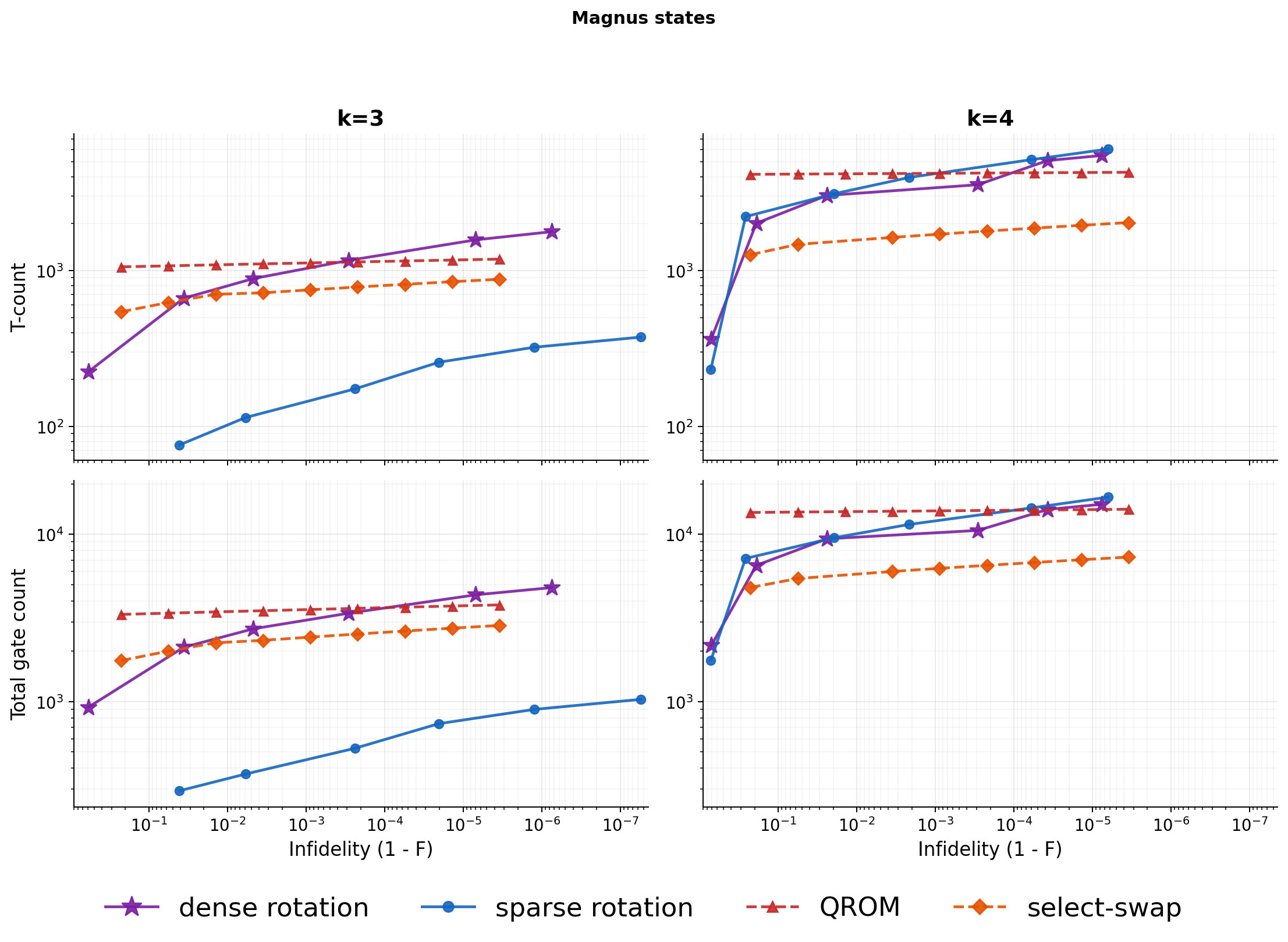}
    \caption{Infidelity versus $T$-count and total gate count for Magnus benchmark states with $k=3$ and $k=4$. The top row reports $T$-count and the bottom row reports total gate count.}
    \label{fig:magnus-inf-vs-counts}
\end{figure}

To compare method scaling on this family, \Cref{fig:magnus-lookup-vs-k} reports $T$-count and total gate count for \QROM, \SelectSwap, and the dense and sparse rotation-based methods. The lookup methods are shown for Magnus states with $k=1,\ldots,6$, while the rotation-based methods are shown for $k=3,\ldots,6$. The three columns use fixed precision parameters $b=11,18,25$, which on the lookup branch were chosen from the target infidelity levels $10^{-3}$, $10^{-5}$, and $10^{-7}$. These permutation-register benchmarks are highly sparse in the computational basis. Under the encoding used here, the support size is exactly $k!$, while the minimal register size is $n=k\lceil \log_2 k\rceil$, so the ambient dimension is $2^n$. Thus the occupied fraction is $k!/2^{k\lceil \log_2 k\rceil}$, which for $k=3,4,5,6$ is $6/64$, $24/256$, $120/32768$, and $720/262144$, respectively. This explains why the sparse rotation method remains competitive on this family and can outperform the dense method for some $k$: although the support itself grows factorially, the encoding becomes increasingly sparse inside the full computational basis. On the lookup branch, \SelectSwap consistently improves on \QROM. More broadly, these benchmarks treat the Magnus states as generic encoded vectors and exploit only sparsity, not additional algebraic structure. Accordingly, the costs reported here should not be interpreted as the intrinsic complexity of Magnus-based simulation; recent structure-aware Magnus constructions can reduce the relevant state-preparation overhead to polynomial cost~\cite{fang2025time,fang2025high}.
% \DL{We may add some discussion about the sparsity here. Explicitly calculate the sparsity of Magnus state and then explain why the sparse method works well for some cases. Also mention that the complexity for preparing magnus state is super-polynomial but a recent paper Magnus paper develop a method with structure aware method to reduce this to poly cost. }

\begin{figure}[htbp]
    \centering
    \includegraphics[width=\textwidth]{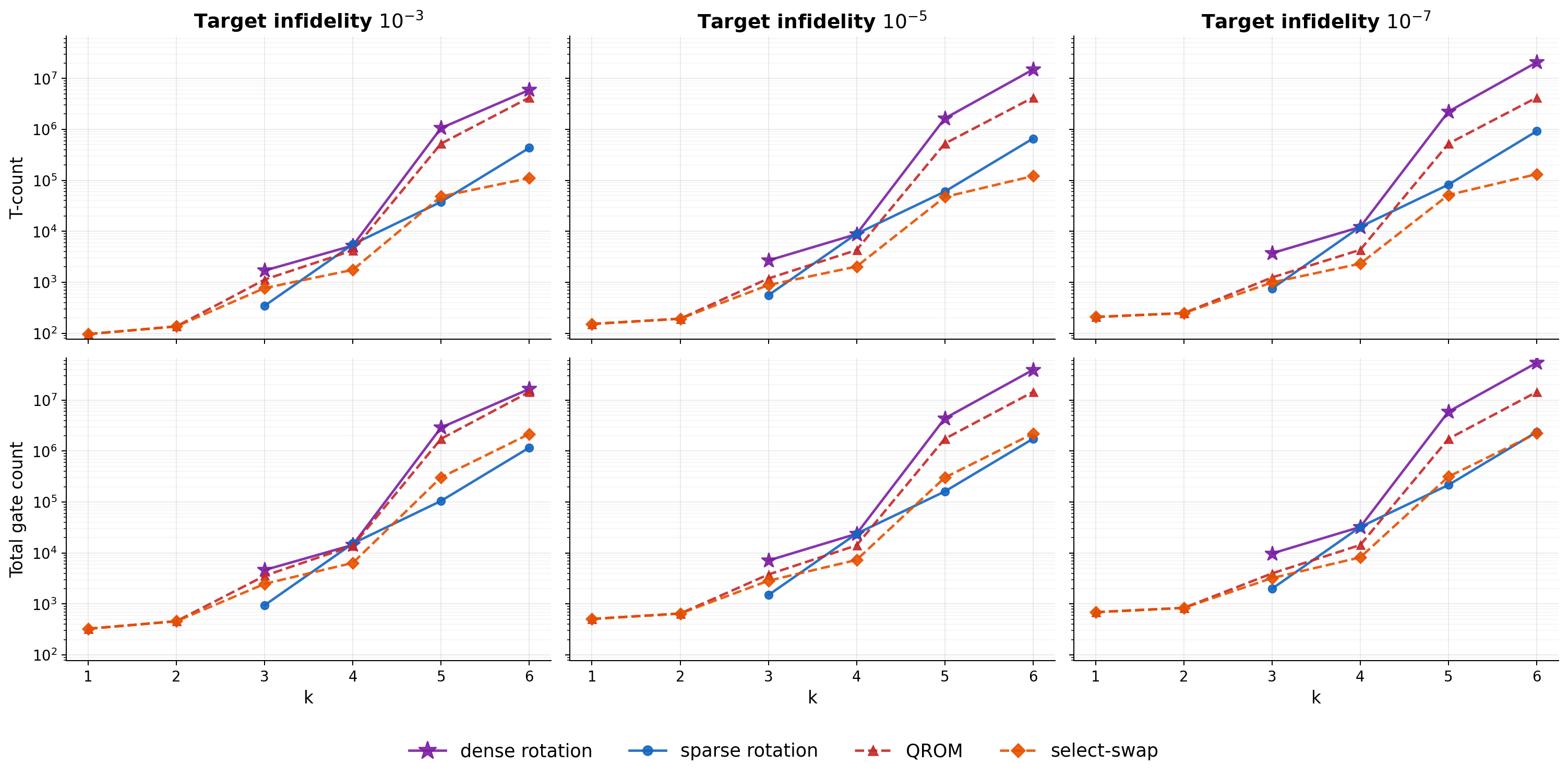}
    \caption{Resource counts versus $k$ for Magnus benchmark states. Columns correspond to the fixed precision settings $b=11,18,25$, selected on the lookup branch from target infidelity levels $10^{-3}$, $10^{-5}$, and $10^{-7}$. \QROM and \SelectSwap are shown for $k=1,\ldots,6$, while the dense and sparse rotation-based methods are shown for $k=3,\ldots,6$. The top row reports $T$-count and the bottom row reports total gate count.}
    \label{fig:magnus-lookup-vs-k}
\end{figure}

\section{Discussion and conclusion}
\label{sec:conclusions}

Quantum state preparation is a basic subroutine in many quantum algorithms, and its practical cost can strongly affect end-to-end algorithmic performance. In this work, we compared representative rotation-based and sampling-based methods for preparing general $n$-qubit states with real amplitudes, with an emphasis on resource measures relevant to fault-tolerant quantum computing. Beyond the usual comparison in terms of CNOT count or $T$-count, we also incorporated total gate count and compilation overhead into the evaluation. Our results show that sampling-based methods achieve asymptotically lower $T$-count and maintain an overall advantage after accounting for total gate count and compilation overhead for a general dense quantum state. Indeed, this aligns well with the theoretical estimate that achieving $\epsilon$-precision state preparation in $\|\cdot\|_{\ell_2}$ requires $L\log_2(L/\epsilon) + \Theta(\sqrt{L})$ total gates. This yields a constant-factor improvement over state-of-the-art rotation-based methods, including the repeat-until-success method with complexity $1.15\, L\log_2(L/\epsilon)$ and Ross--Selinger (gridsynth) synthesis with complexity $3L\log_2(L/\epsilon)$, as well as an asymptotically better scaling than Solovay--Kitaev-type methods. Finally, we developed a software package for compiling state preparation circuits, providing a practical tool for estimating the logical resources and constructing circuits for state-preparation subroutines, which can be incorporated into other quantum algorithms implementation. 

We note that the only structural property considered implicitly in this paper is the sparsity of the target quantum state. However, additional forms of structure may also be exploited in state preparation, including functional, low-rank, and tensor structure. While there has been existing work on state preparation for such structured classes of states, our resource estimates suggest that sampling-based methods are expected to offer lower practical cost while maintaining optimal asymptotical gate complexity, and we expect similar advantages may also arise for these more specialized state classes when combined with new sampling strategies. Developing such sampling methods, together with the corresponding quantum circuit optimization and compilation techniques, is left for future work.

\section*{Software availability}
The state-preparation compilation package used in this work will be made available upon request.

\section*{Acknowledgments}
 D.L., H.W and Z.H gratefully acknowledge helpful discussions during the IPAM workshops \textit{Bridging the Gap Between NISQ and FTQC} and \textit{New Frontiers in Quantum Algorithms for Open Quantum Systems}, which contributed to the development of this paper. D.F., C.I., C.Y., and S.Z. acknowledge the support from the U.S. Department of Energy, Office of Science, Accelerated Research in Quantum Computing Centers, Quantum Utility through Advanced Computational Quantum Algorithms, grant NO. DE-SC0025572. D.F. also acknowledges the National Science Foundation CAREER Award DMS-2438074. D.L. and WA.dJ acknowledge support from the U.S. Department of Energy (DOE) under Contract No. DE-AC02-05CH11231, through the Office of Advanced Scientific Computing Research Accelerated Research for Quantum Computing Program, MACH-Q project. H.W. and J.C. acknowledge support from the NSF Challenge Institute for Quantum Computation with Grant No. 2016245, National Quantum Virtual Laboratory with Grant No. 2410716, and a gift from Google LLC.
 The authors thank Yu Tong, Guang Hao Low, and Leo Zhou for valuable discussions.

\bibliographystyle{unsrt}
\bibliography{bibo}
\end{document}

%% file: ex_shared.tex
\usepackage{lipsum}
\usepackage{amsfonts}
\usepackage{graphicx}
\usepackage{epstopdf}
\usepackage{algorithmic}
\usepackage{xspace}
\ifpdf
  \DeclareGraphicsExtensions{.eps,.pdf,.png,.jpg}
\else
  \DeclareGraphicsExtensions{.eps}
\fi

\newsiamremark{remark}{Remark}
\newsiamremark{hypothesis}{Hypothesis}
\crefname{hypothesis}{Hypothesis}{Hypotheses}
\newsiamthm{claim}{Claim}
\newsiamremark{fact}{Fact}
\crefname{fact}{Fact}{Facts}

\headers{Logical Resource Estimation for Quantum State Preparation}{D. Liu, H. Wang, S.Zhu, et al.}

\title{Logical Resource Estimation for Quantum State Preparation with compilation}

\author{
Diyi Liu\thanks{Lawrence Berkeley National Laboratory, Berkeley, CA
}
\and Hanyu Wang\thanks{Department of Electrical and Computer Engineering, University of California, Los Angeles, CA
}
\and Shuchen Zhu\thanks{Department of Mathematics, Duke University, Durham, NC.}
\and Jason Cong\footnotemark[2]
\and Wibe A. de Jong\footnotemark[1]
\and Di Fang\footnotemark[3]
\and Zhen Huang\thanks{Department of Mathematics, University of California, Berkeley, CA.}
\and Costin Iancu\footnotemark[1]
\and Chao Yang\footnotemark[1]
}

\usepackage{amsopn}

\xspaceaddexceptions{~}
\newcommand{\QROM}{\textnormal{\textsc{QROM}}\xspace}
\newcommand{\QROMOp}{\mathsf{QROM}}
\newcommand{\SelectSwap}{\textnormal{\textsc{SelectSwap}}\xspace}